\documentclass{elsarticle}

\journal{Theoretical Computer Science}

\usepackage{amsmath, amsthm, amssymb}
\usepackage{stmaryrd}

\usepackage{graphicx}
\usepackage{tikz}
\usetikzlibrary{automata,arrows.meta}
\usepackage{subcaption}

\usepackage{algorithm}
\usepackage{algpseudocode}

\usepackage{listings}
\lstnewenvironment{haskell}[1][]
  {\lstset{language=Haskell, basicstyle=\sffamily, columns=fullflexible, keywords={data, deriving, via, class, type, where}, emph={Eq, Ord, Nominal}, emphstyle=\slshape, #1}}{}
\newcommand{\hask}[1]{{\lstinline[language=Haskell, basicstyle=\sffamily, columns=fullflexible, keywords={data, deriving, via, class, type, where}, emph={Eq, Ord, Nominal}, emphstyle=\slshape]+#1+}}

\usepackage[colorlinks=true]{hyperref}

\DeclareMathOperator{\Sym}{Sym}
\DeclareMathOperator{\ext}{ext}
\DeclareMathOperator{\Nsize}{N}
\DeclareMathOperator{\im}{Im}
\newcommand{\Q}{\mathbb{Q}}
\newcommand{\N}{\mathbb{N}}
\newcommand{\lang}{\mathcal{L}}
\newcommand{\Lmax}{\lang_{\text{max}}}
\newcommand{\Lint}{\lang_{\text{int}}}

\newcommand{\cplusplus}{C\texttt{++}}

\newcommand{\ONS}{\textsc{Ons}}
\newcommand{\ONShs}{\textsc{Ons-hs}}
\newcommand{\ONSboth}{\textsc{Ons(-hs)}}
\newcommand{\NLambda}{\textsc{N}$\lambda$}
\newcommand{\LOIS}{\textsc{Lois}}

\newcommand{\LStar}{\textsc{L}$^\star$}
\newcommand{\nLStar}{$\nu$\LStar}

\newcommand{\extendset}[1]{\llbracket #1 \rrbracket}
\newcommand{\extendproduct}[1]{\llbracket #1 \rrbracket}

\newtheorem{theorem}{Theorem}[section]
\newtheorem{lemma}[theorem]{Lemma}
\newtheorem{corollary}[theorem]{Corollary}
\newtheorem{proposition}[theorem]{Proposition}
\theoremstyle{definition}
\newtheorem{example}[theorem]{Example}
\newtheorem{remark}[theorem]{Remark}
\newtheorem{definition}[theorem]{Definition}

\begin{document}

\begin{frontmatter}

\title{Fast Computations on Ordered Nominal Sets\tnoteref{mytitlenote}}

\tnotetext[mytitlenote]{
This is a revised and extended version of a paper which appeared in the proceedings
    of ICTAC 2018~\cite{VenhoekMR18}.
    The research leading to these results has received funding from the European Union's Horizon 2020 research and innovation programme under the Marie Sk\l{}odowska-Curie Grant Agreement No. 795119 and the ERC AdG project 787914 FRAPPANT.
}

\author[run]{David Venhoek}
\ead{david@venhoek.nl}
\author[ou]{Joshua Moerman}
\ead{joshua.moerman@ou.nl}
\author[run]{Jurriaan Rot}
\ead{jrot@cs.ru.nl}

\address[run]{Institute for Computing and Information Sciences,\\
Radboud Universiteit, Nijmegen, The Netherlands}
\address[ou]{Open Universiteit, Heerlen, The Netherlands}

\begin{abstract}
Nominal automata are models for recognising
languages over infinite alphabets,
based on the
algebraic notion of \emph{nominal set}.
Motivated by their use in automata theory,
we show how to compute efficiently with nominal sets over the
so-called \emph{total order symmetry}, a variant which allows to compare alphabet
letters for equality as well as their respective order.
We develop an explicit finite representation of such nominal sets
and basic constructions thereon.
The approach is implemented
as the library \ONS{} (Ordered Nominal Sets),
enabling programming with infinite sets.
Returning to our motivation of nominal automata, we evaluate \ONS{} in two
applications: minimisation of automata
and active automata learning. In both cases, \ONS{} is competitive compared to
existing implementations and outperforms them for certain classes of inputs.
\end{abstract}


\end{frontmatter}

\section{Introduction}

Automata over infinite alphabets are natural models for
programs with unbounded data domains.
Such automata, often formalised as \emph{register automata}~\cite{KaminskiF94}, are applied in the modelling and analysis of communication protocols, hardware, and software systems
(see~\cite{bojanczyk2014automata,DAntoniV17,DBLP:journals/corr/abs-1209-0680,KaminskiF94,MontanariP97,Segoufin06} and references therein).
Typical infinite alphabets include sequence numbers, timestamps and identifiers.
This means that one can model data flow in such automata beyond the basic control flow provided by ordinary automata.
Recently, it has been shown 
in a series of papers that such models are amenable to
learning~\cite{AartsFKV15,BolligHLM13,CasselHJS16,DrewsD17,moerman2017learning,Vaandrager17}
with the verification of (closed source) implementations of the Transmission Control Protocol (TCP) as a prominent example~\cite{Fiterau-Brostean16}.

A foundational approach to
infinite alphabets is provided by the algebraic notion of
\emph{nominal set}. In computer science, nominal sets were originally introduced 
as an elegant formalism for name binding~\cite{GabbayP02,pitts2016nominal}.
They have been used in a variety of applications in
semantics, computation and concurrency theory; see~\cite{pitts2013nominal} for an overview. 

In particular, Boja{\'n}czyk et al.\ introduce \emph{nominal automata}, which allow one to model languages over infinite
alphabets modelled as nominal sets~\cite{bojanczyk2014automata}.
Nominal automata are defined as ordinary automata by replacing finite sets with ``orbit-finite nominal sets''.
This means that the state space of a nominal automaton is finite up to renaming of the infinite data occuring in the alphabet. 
As a consequence of orbit-finiteness, one can represent these automata finitely and compute with them---for instance,
emptiness and equivalence of deterministic automata can be decided with a slight adaptation of the classical algorithms. 
Nominal automata are equally expressive as register automata: the connection between these models is well
exposed by Boja\'{n}czyk \cite{Bojanczyk18}. Nominal automata thereby provide an elegant mathematical foundation
of register automata and automata over infinite alphabets.

Important for applications of nominal sets and automata are implementations.
A couple of tools exist to compute with nominal sets.
Notably, \NLambda{} \cite{EPTCS207.3} and \LOIS{} \cite{kopczynski1716,kopczynski2017} provide a general purpose programming language to manipulate infinite sets.%
\footnote{Other implementations of nominal techniques that are less directly related to our setting (Mihda, Fresh OCaml and Nominal Isabelle) are discussed in Section~\ref{sec:related}.}
Both tools are based on SMT solvers and use logical formulas to represent the infinite sets.
These implementations are very flexible, and the SMT solver does most of the heavy lifting, which makes the implementations themselves relatively straightforward.
Unfortunately, this comes at a cost as SMT solving is in general {\scshape PSPACE}-hard. 
Since the formulas used to describe sets tend to grow as more calculations are done, running times can become unpredictable.

The overall aim of the current paper is to provide a library
for computing with nominal sets that avoids this implicit representation 
and allows for a precise complexity analysis of algorithms based on it.
The key idea is to represent a nominal set explicitly in terms of its orbits,
which then also provides a natural notion of size---e.g., in a nominal automaton, the number of states is typically
taken to be the number of orbits of the state space. 
Towards this aim, we will focus on a specific variant of nominal sets, which we recall first. 

The approach of Boja{\'n}czyk et al.\ is based on nominal sets over different ``symmetries'',
which stipulate the way in which data values can be compared. As a consequence,
their results on nominal automata are parametric in the structure 
of the data values. Two important such structures on data values are:
\begin{itemize}
	\item
	Data values that can only be compared for equality, referred to as nominal sets
	over the \emph{equality symmetry}.

	\item
	Ordered data values, where data values are rational numbers and can be compared for their order; referred to as the \emph{total order symmetry}.
\end{itemize}

Most of the classical theory and applications of nominal sets are based on the equality symmetry.
In particular, this is the data domain which is most relevant to capture name binding
and it was the first data domain considered for register automata.

In this paper, however, we focus specifically on nominal sets over the total order symmetry,
for the following reasons.
First, the total order symmetry is of interest on its own: allowing to compare data
values allows for interesting examples such as data structures (priority queues and search trees use order)~\cite{CasselHJS16} and data nets with a dense linear order~\cite{LazicNORW08}.
Second, it is a natural example of a \emph{homogeneous structure}, which is a property needed for many decidability problems.
For instance, various problems on data nets are decidable for (strongly) homogeneous structures~\cite{LasotaP18}, and automata learning of register automata is feasible for $\infty$-extendable words~\cite{CasselHJS16} (a notion related to homogeneity).
Without homogeneity, these results are not know to hold.
Third, nominal sets over the total order symmetry are fairly general: they
subsume nominal sets over the equality symmetry. Note however
that representing nominal sets over the equality symmetry this way
requires more orbits in general. 
Fourth, as we will show, nominal sets over the total order symmetry allow for a remarkably simple representation.
The key insight is that the representation of nominal sets from~\cite{bojanczyk2014automata} (and also~\cite{CianciaKM10}) simplifies as follows;
the ``local symmetries'' are trivial, so that each orbit is presented solely by a natural number, indicating the number of variables or registers.

\noindent \\
Our main contributions include the following.
\begin{itemize}
	\item
	We develop the \emph{representation theory} of nominal sets over the total order symmetry.
	We give concrete representations of nominal sets, their products and equivariant maps.
	\item
	We provide \emph{time complexity bounds}
	for operations on nominal sets such as intersections and membership.
	Using those results we give the time complexity of Moore's minimisation algorithm (generalised to nominal automata) and prove that it is polynomial in the number of orbits.
	\item
	Using the representation theory, we are able to \emph{implement nominal sets in a \cplusplus{} library} \ONS{}.
	The library includes all the results from the representation theory (sets, products and maps).

	We also developed a Haskell implementation, called \ONShs{}.
	This allows us to give a more comprehensive evaluation than with the \cplusplus{} implementation alone, avoiding results that are caused by specific implementation details.
	Further, the Haskell implementation is \emph{generic}, meaning that nominal computations are possible even with custom data types.
	
	\item
	We \emph{evaluate the performance} of \ONSboth{}, and compare it to \NLambda{} and \LOIS{},
	using two algorithms on nominal automata: minimisation~\cite{BojanczykL12} and 
	automata learning \cite{moerman2017learning}.
	We use randomly generated automata as well as concrete, logically structured models such as FIFO queues.
	For random automata, our methods are considerably faster in most cases than the other tools.
	On the other hand, \LOIS{} and \NLambda{} are faster in minimising the structured automata as they exploit their logical structure.
	In automata learning, the logical structure is not available a-priori,
	and \ONSboth{} is faster in most cases.
\end{itemize}

The structure of the paper is as follows. 
The first three sections contain background material:
Section~\ref{sec:nomsets} on nominal sets,
Section~\ref{sec:aut} on nominal automata, and 
Section~\ref{subsec:nomorbit} on representation of nominal sets. 
Next, Section~\ref{sec:total-order}
describes the concrete representation of nominal sets, equivariant
maps and products in the total order symmetry. 
The implementations in \cplusplus{} and Haskell are presented in Sections~\ref{sec:impl} and~\ref{sec:haskell} respectively.
Complexity results are presented in Section~\ref{sec:complexity}. 
Section~\ref{sec:evaluation} reports on the evaluation of \ONS{} on algorithms for nominal automata. 
Related work is discussed in Section~\ref{sec:related},
and future work in Section~\ref{sec:fw}.

The current paper extends the conference version (ICTAC 2018~\cite{VenhoekMR18}) 
with proofs of all results, new experiments for evaluating~\ONS{} based on randomly generated formulas, 
and an implementation in Haskell,~\ONShs{}.

\section{Nominal sets}\label{sec:nomsets}

Nominal sets are infinite sets that carry certain symmetries, allowing a finite representation
in many interesting cases. We recall their
formalisation in terms of group actions, 
following~\cite{bojanczyk2014automata,pitts2013nominal},
to which we refer for an extensive introduction.

\subsection{Group actions.}\label{sec:g-sets}

Let $G$ be a group, with the multiplication denoted by juxtaposition
and the unit by $1$. 
Given a set $X$,
a \emph{(right) $G$-action} is a function ${\cdot}\colon X \times G \to X$ satisfying $x \cdot 1 = x$ and $(x \cdot g) \cdot h = x \cdot (g h)$ for all $x \in X$ and $g, h \in G$.
A set $X$ with a $G$-action is called a \emph{$G$-set} and we often write $xg$ instead of $x \cdot g$.
The \emph{orbit} of an element $x \in X$ is the set $\{xg \mid g \in G\}$.
A $G$-set is always a disjoint union of its orbits (in other words, the orbits partition the set).
We say that $X$ is \emph{orbit-finite} if it has finitely many orbits, and we denote the number of orbits by $\Nsize(X)$.

A map $f \colon X\rightarrow Y$ between $G$-sets is called 
\emph{equivariant} if it preserves the group action, i.e., for all $x\in X$ and $g \in G$ we have $f(xg) = f(x)g$.
If an equivariant map $f$ is bijective, then $f$ is an \emph{isomorphism} and we write $X \cong Y$.
A subset $Y \subseteq X$ is called \emph{equivariant} if for all $y \in Y$ and $g \in G$, we have $y g \in Y$. 
The \emph{product} of two $G$-sets $X$ and $Y$ is given by the Cartesian
product $X \times Y$ with the pointwise group action on it, i.e., $(x,y)g = (xg,yg)$.
Union and intersection of $X$ and $Y$ are well-defined if the two
actions agree on their common elements. 

\subsection{Nominal sets.} 

A \emph{data symmetry} is a pair $(\mathcal{D},G)$ where $\mathcal{D}$
is a set and $G$ is a subgroup of $\Sym(\mathcal{D})$, the group 
of bijections on $\mathcal{D}$.
Note that the group $G$ acts on $\mathcal{D}$ 
by defining $x g = g(x)$.\footnote{This is a well-defined action if we use the group multiplication 
$f \cdot g = g \circ f$.}
In the most studied instance, called the \emph{equality symmetry}, 
$\mathcal{D}$ is a countably infinite set and $G = \Sym(\mathcal{D})$. In this paper, 
we will mostly focus on the \emph{total order symmetry}
given by $\mathcal{D} = \mathbb{Q}$ and $G = \{\pi \mid \pi \in \Sym(\mathbb{Q}), \pi \text{ is monotone}\}$.

Let $(\mathcal{D}, G)$ be a data symmetry and $X$ be a $G$-set.
A finite set of data values $S\subseteq \mathcal{D}$ is called a \emph{support} of an element $x \in X$ 
if for all $g\in G$ with $\forall s \in S\colon sg = s$ we have $xg = x$.
A $G$-set $X$ is called \emph{nominal} if every element $x\in X$ has a (necessarily finite) support.

\begin{example}
We list several examples for the total order symmetry $(\mathbb{Q}, G)$.
The set $\mathbb{Q}^2$ is a $G$-set, with the action defined by
$(q_1, q_2) g = (g(q_1), g(q_2))$ for all $g \in G$.  
It is nominal, as each element  $(q_1, q_2) \in \mathbb{Q}^2$ has the finite set $\{q_1, q_2\}$ as a support.
Indeed, if a monotone bijection $g \in G$ is the identity on $q_1$ and $q_2$ then we have
$(q_1, q_2) g = (g(q_1), g(q_2)) = (q_1, q_2)$ as required.
The set $\mathbb{Q}^2$ has the following three orbits: 
\[
\{(q_1, q_2) \mid q_1 < q_2\}\,, ~
\{(q_1, q_2) \mid q_1 > q_2\} \,, ~
 \{(q_1, q_2) \mid q_1 = q_2\} \,.
\]
To see this, note that each of these sets is closed under monotone permutations, and
their union is the full set $\mathbb{Q}^2$. 

For a set $X$, the set of all subsets of size $n \in \mathbb{N}$ is denoted by
\[ \mathcal{P}_n(X) = \{Y \subseteq X \mid |Y| = n\} \,. \]
The set $\mathcal{P}_n(\mathbb{Q})$ is a single-orbit nominal set for each $n$, with the action
defined by direct image: $Y g = \{ yg \mid y \in Y \}$. Every set $Y \in \mathcal{P}_n(X)$
is a support of itself. 

The group of monotone bijections also acts by direct image on the full power set $\mathcal{P}(\mathbb{Q})$,
but this is \emph{not} a nominal set. For instance, the set $\mathbb{Z} \in \mathcal{P}(\mathbb{Q})$ of integers 
has no finite support.
\end{example}

If $S\subseteq \mathcal{D}$ is a support of an element $x \in X$, then any finite set $S'\subseteq \mathcal{D}$ such that $S\subseteq S'$ is also a support of $x$.
A set $S\subseteq \mathcal{D}$ is a \emph{least support} of $x \in X$ if it is a support of $x$ and 
$S \subseteq S'$ for any support $S'$ of $x$.
The existence of least supports is crucial for representing orbits. 
Unfortunately, even when elements have a finite support, in general they do not always have a least support.
A data symmetry $(\mathcal{D},G)$ is said to \emph{admit least supports} if every element of every nominal set 
has a least support. Both the equality and the total order symmetry admit least supports.
For other (counter)examples of data symmetries admitting least supports, see~\cite{bojanczyk2014automata}.
Having least supports is useful for a finite representation.

Given a nominal set $X$, the size of the least support of an element $x \in X$ is denoted by $\dim(x)$, the \emph{dimension} of $x$.
We note that all elements in the orbit of $x$ have the same dimension.
For an orbit-finite nominal set $X$, we define 
\[ \dim(X) = \max \{ \dim(x) \mid x \in X \} \,. \]
For a single-orbit nominal set $O$, observe that $\dim(O) = \dim(x)$ where $x$ is any element $x \in O$.

\section{Automata over Nominal Sets}
\label{sec:aut}

In this section we recall the notion of nominal automata, which allow to recognise languages
over infinite alphabets. 
Nominal automata are defined as ordinary automata by replacing the finite state space and finite alphabet
by orbit-finite nominal sets.

The theory of nominal automata is developed in~\cite{bojanczyk2014automata}, where it is shown that many algorithms from automata theory transfer to nominal automata.
For instance, emptiness and equivalence of deterministic automata can be decided with a slight adaptation of the classical algorithms.
Nonetheless, not all algorithms generalise: equivalence of non-deterministic automata is undecidable in the nominal setting.
Below, we also briefly describe minimisation and learning of nominal automata. 
We start with an introductory example of a language over the rational numbers. 

\begin{example}\label{ex:int-lang}
Consider the following language on rational numbers:
\[ \Lint = \{ a_1 b_1 \cdots a_n b_n \mid a_i, b_i \in \Q, a_i < a_{i+1} < b_{i+1} < b_i \text{ for all } i \}. \]
We call this language the \emph{interval language} as a word $w \in \Q^{\ast}$ is in the language when it denotes a sequence of nested intervals.
For instance, the word $w = 0\, 2\, 1\, \frac{3}{2}$ is in $\Lint$, but $w' = 0\, 2\, 1\, 1$ is not.
This language contains arbitrarily long words.
For this language it is crucial to work with an infinite alphabet as for each finite set $C \subset \Q$, the restriction $\Lint \cap C^{\ast}$ is just a finite language.
Note that the language is equivariant: $w \in \Lint \iff wg \in \Lint$ for any monotone bijection $g$, because nested intervals are preserved by monotone maps.%
\footnote{The $G$-action on words is defined point-wise: $(w_1 \ldots w_n) g = (w_1 g) \ldots (w_n g)$.}
Indeed, $\Lint$ is a nominal set, although it is not orbit-finite.

Informally, the language $\Lint$ can be accepted by the automaton depicted in Figure~\ref{fig:lint-minimal}.
Here we allow the automaton to store rational numbers and compare them to new symbols.
The input on a transition is denoted by the variable $a, b$ or $c$ which does not already appear in the source state.
This input may be constrained by a guard.
For example, the transition from $q_2$ to $q_3$ is taken if any value $c$ between $a$ and $b$ is read and then the currently stored value $a$ is replaced by $c$.
For any other value read at state $q_2$ the automaton transitions to the sink state $q_4$.
Such a transition structure is made precise by the notion of nominal automaton.

\begin{figure}
\centering
\scalebox{.85}{
\begin{tikzpicture}[-{Latex}]
\node (S 0) at (0,0) [state, minimum size=1.5cm] {$q_0$};
\node (S 1) at (3,0) [state, minimum size=1.5cm] {$q_1(a)$};
\node (S 2) at (7,0) [state, minimum size=1.5cm, accepting] {$q_2(a,b)$};
\node (S 3) at (11,0) [state, minimum size=1.5cm] {$q_3(a,b)$};
\node (S 4) at (7,-2.5) [state, minimum size=1.5cm] {$q_4$};
\node (S 5) at (11,-2.5) [state,minimum size=1.5cm,draw=none] {};

\path
(S 0) edge node[above] {$a$} (S 1)
(S 1) edge node[above] {$b > a$} (S 2)
(S 1) edge[bend right] node[below left] {$b \le a$} (S 4)
(S 2) edge[bend left=10] node[above,align=center] {$a<c<b$\\$a\leftarrow c$} (S 3)
(S 3) edge[bend left=10] node[below,align=center] {$a<c<b$\\$b\leftarrow c$} (S 2)
(S 2) edge[bend right=15] node[left] {$c\le a$} (S 4)
(S 2) edge[bend left=15] node[right] {$c\ge b$} (S 4)
(S 3) edge[bend left] node[above left] {$c\le a$} (S 4)
(S 3) edge[bend left=60] node[below right] {$c\ge b$} (S 4);
\draw[-{Latex},rounded corners=5] (S 4.south west) -- ([yshift=-.5cm]S 4.south west) -- node[midway,below]{$a$} ([yshift=-.5cm]S 4.south east) -- (S 4.south east);
\end{tikzpicture}
}
\caption{Example automaton that accepts the language $\Lint$.}
\label{fig:lint-minimal}
\end{figure}
\end{example}

\begin{remark}
For the reader familiar with \emph{register automata}, Figure~\ref{fig:lint-minimal} resembles a register automaton.
One notable difference to the classical definition of register automaton is that we allow each state to have a different set of registers.
In this example, $q_0$ and $q_4$ have no registers, $q_1$ has a single register and the other states have two registers.
This way, the values in the registers are always defined.

In general, register automata and nominal automata are equally expressive; a clear
comparison is given by Boja{\'n}czyk~\cite{Bojanczyk18}. Each state in a register automaton
gives rise to one or more orbits in an equivalent nominal automaton. In fact, the number of orbits
may be exponential in the number of states of the register automaton. These orbits represent distinct
relations between registers.
\end{remark}

\begin{definition}
A \emph{nominal language} is an equivariant subset $L \subseteq A^{*}$ where $A$ is an orbit-finite nominal set.
\end{definition}

\begin{definition}
A \emph{nominal deterministic finite automaton} is a tuple $(S,A,F,\delta)$, where $S$ is an orbit-finite nominal set of states, $A$ is an orbit-finite nominal set of symbols, $F \subseteq S$ is an equivariant subset of final states and $\delta \colon S \times A \to S$ is the equivariant transition function.

Given a state $s \in S$, we define the usual acceptance condition: a word $w \in A^{*}$ is \emph{accepted} if $w$ denotes a path from $s$ to a final state.
\end{definition} 

The automaton in Figure~\ref{fig:lint-minimal} can be formalised as a nominal deterministic finite automaton as follows.
Let
\[ S = \{ q_0, q_4 \} \cup \{ q_1(a) \mid a \in \Q \} \cup \{ q_2(a, b) \mid a < b \in \Q \} \cup \{ q_3(a, b) \mid a < b \in \Q \} \]
be the set of states, where the group action is defined as one would expect.
The transition we described earlier can now be defined formally as
\[ \delta(q_2(a,b), c) = q_3(c,b) \quad \text{for all } a < c < b \in \Q \,. \]
By defining $\delta$ on all states accordingly and defining the final states as
\[ F = \{ q_2(a, b) \mid a < b \in \Q \} \,, \]
we obtain a nominal deterministic automaton $(S, \Q, F, \delta)$.
The state $q_0$ accepts the language $\Lint$.

\subsection{Minimisation of Nominal Automata}
\label{sec:minimisation}

For languages recognised by nominal DFAs, a Myhill-Nerode theorem holds which relates states to right congruence classes \cite{bojanczyk2014automata}.
This guarantees the existence of unique minimal automata.
We say an automaton is \emph{minimal} if its set of states has the least number of orbits and each orbit has the smallest dimension possible.%
\footnote{Abstractly, an automaton is minimal if it has no proper quotients. Minimal deterministic automata are unique up to isomorphism.}

\begin{example}\label{ex:minimal-automaton}
Consider the language 
\[ \Lmax = \{ w a \in \Q^{*} \mid a = \max(w_1, \dots, w_n) \} \]
consisting of those words where the last symbol is the maximum of previous symbols.
Figure~\ref{fig:lmax-non-minimal} depicts a nominal automaton accepting $\Lmax$,
which is however not minimal.
Figure~\ref{fig:lmax-minimal} is the minimal nominal automaton accepting $\Lmax$.
\begin{figure}
	\centering
	\scalebox{0.8}{
		\begin{tikzpicture}[-{Latex}]
		\node (S 0) at (0, 0) [state, minimum size=1.5cm] {$q_0$};
		\node (S 1) at (3, 0) [state, minimum size=1.5cm] {$q_1(a)$};
		\node (S 2) at (7.5, 0) [state, minimum size=1.5cm, accepting] {$q_2(a,a)$};
		\node (S 3) at (7.5, 2.5) [state, minimum size=1.5cm] {};
		\node at (7.5, 2.5) {\begin{tabular}{c}$q_3(a,b)$\\ $b > a$\end{tabular}};
		\node (S 4) at (7.5, -2.5) [state, minimum size=1.5cm] {};
		\node at (7.5,-2.5) {\begin{tabular}{c}$q_4(a,b)$\\ $b < a$\end{tabular}};

		\draw ([xshift=-0.5cm]S 0.west) -- (S 0.west);
		\path (S 0) edge node[midway,above] {$a$} (S 1);
		\path (S 1) edge [bend left] node[midway,sloped,above] {$b>a$} (S 3);
		\path (S 1) edge node[midway,above] {$a$} (S 2);
		\path (S 1) edge [bend right] node[midway,sloped,below] {$b<a$} (S 4);
		\path (S 2) edge[bend left] node[midway, left] {$b>a$} (S 3);
		\path (S 3) edge[bend left] node[midway,right] {\begin{tabular}{l}$b$\\ $a\leftarrow b$\end{tabular}} (S 2);
		\path (S 2) edge[bend left] node[midway,right] {$b<a$} (S 4);
		\path (S 4) edge[bend left] node[midway, left] {$a$} (S 2);
		\path (S 3.north west) edge[bend left=80, looseness=2.2] node[midway,above]{\begin{tabular}{c}$c>b$\\$a\leftarrow b$\\$b\leftarrow c$\end{tabular}} (S 3.north east);
		\path (S 4.south east) edge[bend left=80, looseness=2.2] node[midway,below]{\begin{tabular}{c}$c<a$\\$a\leftarrow a$\\$b\leftarrow c$\end{tabular}} (S 4.south west);
		\path (S 2.north east) edge[bend left=80, looseness=2.2] node[midway,right]{$a$} (S 2.south east);
		\path (S 3) edge[bend left=90, looseness=2] node[midway,right]{\begin{tabular}{r}$c<b$\\$a\leftarrow b$\\$b\leftarrow c$\end{tabular}} (S 4);
		\path (S 4) edge[bend right=80, looseness=1.9] node[midway,left]{\begin{tabular}{l}$c>a$\\$a\leftarrow a$\\$b\leftarrow c$\end{tabular}} (S 3);
		\end{tikzpicture}
	}
	\caption{Example automaton that accepts the language $\Lmax$.}
	\label{fig:lmax-non-minimal}
\end{figure}

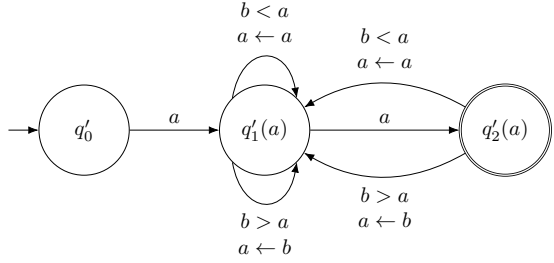
\begin{figure}
\centering
\scalebox{0.8}{
\begin{tikzpicture}[-{Latex}]
\node (S 0) at (0,0) [state, minimum size=1.5cm] {$q'_0$};
\node (S 1) at (3,0) [state, minimum size=1.5cm] {$q'_1(a)$};
\node (S 2) at (7,0) [state, accepting, minimum size=1.5cm] {$q'_2(a)$};
\draw[-{Latex}] ([xshift=-0.5cm]S 0.west) -- (S 0.west);
\draw[-{Latex}] (S 0.east) -- node[midway,above] {$a$} (S 1.west);
\draw[-{Latex}] (S 1.east) -- node[midway,above] {$a$} (S 2.west);
\path (S 1.north west) edge[bend left=80, looseness=2.2] node[midway,above] {\begin{tabular}{c}$b<a$\\$a\leftarrow a$\end{tabular}} (S 1.north east);
\path (S 1.south west) edge[bend right=80, looseness=2.2] node[midway,below] {\begin{tabular}{c}$b>a$\\$a\leftarrow b$\end{tabular}} (S 1.south east);

\path (S 2) edge [bend right] node[midway,above] {\begin{tabular}{c}$b<a$\\$a\leftarrow a$\end{tabular}} (S 1);
\path (S 2) edge [bend left] node[midway,below] {\begin{tabular}{c}$b>a$\\$a\leftarrow b$\end{tabular}} (S 1);
\end{tikzpicture}
}
\caption{The automaton from Figure~\ref{fig:lmax-non-minimal}, minimised.}
\label{fig:lmax-minimal}
\end{figure}
\end{example}

There are several algorithms for minimising deterministic automata.
In this paper we focus on Moore's minimisation algorithm.
It generalises to nominal DFAs since it uses set operations which work just as well on nominal sets (see Algorithm~\ref{alg:moore}).
We will perform a complexity analysis in Section~\ref{sec:complexity} and later use this algorithm for testing our library.

\begin{algorithm}
\caption{Moore's minimisation algorithm for nominal DFAs}\label{alg:moore}
\begin{algorithmic}[1]
\Require{ Nominal automaton $(S,A,F,\delta)$.}
\State{$i \leftarrow 0$, ${\equiv_{-1}} \leftarrow S\times S$, ${\equiv_{0}} \leftarrow F\times F \cup (S\backslash F)\times (S\backslash F)$}
\While{${\equiv_i} \,\ne\, {\equiv_{i-1}}$}
\State{${\equiv_{i+1}} \leftarrow \{(q_1, q_2) \mid q_1 \equiv_i q_2 \,\wedge\, \forall a\in A, \delta(q_1,a) \equiv_i \delta(q_2,a) \}$}
\State{$i \leftarrow i+1$}
\EndWhile
\State{$E\leftarrow S/_{\equiv_i}$}
\State{$F_E \leftarrow \{e\in E \mid \forall s\in e, s\in F\}$}
\State{Let $\delta_E$ be the map such that, if $s\in e$ and $\delta(s,a)\in e'$, then $\delta_E(e,a) = e'$.}
\State{\Return{$(E,A,F_E,\delta_E)$.}}
\end{algorithmic}
\end{algorithm}

\subsection{Learning nominal automata}

Another interesting application is \emph{automata learning}.
The aim of automata learning is to infer an unknown regular language $\lang$.
We use the framework of active learning as set up by Dana Angluin \cite{Angluin87}
where a learning algorithm can query an oracle to gather information about $\lang$.
Formally, the oracle can answer two types of queries:
\begin{enumerate}
\item \emph{membership queries}, where a query consists of a word $w \in A^{*}$ and the oracle replies whether $w \in \lang$, and
\item \emph{equivalence queries}, where a query consists of an automaton $\mathcal{H}$ and the oracle replies positively if $\lang(\mathcal{H}) = \lang$ or provides a counterexample if $\lang(\mathcal{H}) \neq \lang$.
\end{enumerate}
With these queries, the \LStar{} algorithm can learn regular languages efficiently \cite{Angluin87}.
In particular, it learns the unique minimal automaton for $\lang$ using only finitely many queries.
The \LStar{} algorithm has been generalised to \nLStar{} in order to learn \emph{nominal} regular languages \cite{moerman2017learning}.
In particular, it learns a nominal DFA (over an infinite alphabet) using only finitely many queries.
The algorithm is not polynomial, unlike the minimisation algorithm described above.
However, the authors of \emph{loc. cit.} conjecture that there is a polynomial algorithm.\footnote{See \url{joshuamoerman.nl/papers/2017/17popl-learning-nominal-automata.html} for a sketch of the polynomial algorithm.}
For the correctness, termination and comparison with other learning algorithms see \emph{loc. cit.}

Learning register automata is an active research area with applications such as bug-finding in internet protocols \cite{Fiterau-Brostean18}; see~\cite{Vaandrager17} for other applications.
We will implement \nLStar{} to test our library in Section~\ref{sec:learning-results}.

\section{Representing nominal sets}\label{subsec:nomorbit}

In this section we sketch the representation of nominal sets according to~\cite{bojanczyk2014automata}.
We will not need all the details, but we include it to give a rough idea of the used techniques,
and to fix notation.
As noted in Section~\ref{sec:g-sets}, a $G$-set is always a disjoint union of its orbits.
In order to represent an orbit-finite $G$-set, it is hence sufficient to represent each orbit.

An orbit $X$ of a $G$-set can always be represented as a quotient of $G$.
To see this, pick some element $x \in X$ and consider the subgroup $H = \{ g \in G \mid x g = x \}$.
The set of cosets $G / H$ is acted on by $G$ via right-multiplication.
One can check that $X$ is isomorphic to $G / H$ as $G$-sets (for the isomorphism, send $x$ to the equivalence class of $1$).
So an orbit is defined by some subgroup $H$.

What is left is to represent the data $H$ is some finite way.
This is where we need the fact that $x$ has a least support $C$.
We use $C$ to define two subgroups of $H$.
First, there are elements which leave each $c \in C$ fixed, but may move elements outside of $C$.
Second, there are the ``local symmetries'' that act as the identity outside of $C$, but may permute the elements of $C$.
The subgroup $H$ is generated by these two subgroups.
We only need to know the local symmetries $S$.
The whole reduction from an abstract orbit to this subgroup is stated in the following theorem from~\cite{bojanczyk2014automata}.

\begin{theorem}\label{representation_tuple}
Let $X$ be a single-orbit nominal set for a data symmetry 
$(\mathcal{D}, G)$ that admits least supports
and let $C \subseteq \mathcal{D}$ be the least support of some element 
$x \in X$.
Then there exists a subgroup $S \le G|_C$ such that $X \cong [C,S]^{ec}$,
\end{theorem}

\noindent
where we used the following notation.
The \emph{restriction} of a group $G$ to a subset $C \subseteq \mathcal{D}$ is defined as
\[ G|_C \,=\, \{\pi|_C \mid \pi \in G,\, C\pi = C\} \,, \]
where $\pi|_C$ is the restriction of the bijection $\pi \colon \mathcal{D} \to \mathcal{D}$ to the domain $C$.
The \emph{extension} of a subgroup (of local symmetries) $S \le G|_C$ is defined as
\[ \ext_G(S) \,=\, \{\pi \in G \mid \pi|_C\in S\} \,. \]
Finally, for $C \subseteq \mathcal{D}$ and $S \le G|_C$, we define the orbit of right-cosets
\[ [C,S]^{ec} \,=\, G / \ext_G(S) \,=\, \{\{sg \mid s \in \ext_G(S)\} \mid g\in G\} \,. \]

\section{Representation in the total order symmetry}
\label{sec:total-order}

This section develops a concrete representation (based on Theorem~\ref{representation_tuple}) of nominal sets over the total 
order symmetry, as well as equivariant maps and products.
From now on, by \emph{nominal set} we always refer to a nominal set over the total order symmetry. Hence, our data domain is $\Q$ and we take $G$ to be the group of monotone bijections.

\subsection{Orbits and nominal sets}\label{sec:orbits-and-nominal-set}

Our first observation is that the finite group of local symmetries $S$ in Theorem~\ref{representation_tuple} is always trivial, i.e., $S$ is the trivial group, $I=\{1\}$.
This follows from the following lemma and the fact that $S \le G|_C$.

\begin{lemma}\label{group_trivial}
For every finite subset $C \subset \mathbb{Q}$, we have $G|_C = I$.
\end{lemma}
\begin{proof}
Let $\pi \in G|_C$ be any element of $G|_C$. If $\pi$ is not the identity, then since $C$ is finite, there exists a smallest element $c\in C$ with $c\pi \ne c$. 
Since $\pi$ is a bijection mapping $C$ to $C$, we find $c\pi\pi \ne c\pi$ and $c\pi \in C$, hence $c < c\pi$. Furthermore, there exists some $c'\in C$ with $c'\pi = c$. Since by assumption $c' \ne c$, also $c < c'$. But then both $c < c'$ and $c\pi > c = c'\pi$, contradicting monotonicity of $\pi$.
Hence $\pi$ is the identity element, and $G|_C = I$.
\end{proof}

\noindent Immediately, we see that $(C,S) = (C,I)$, and hence that the orbit is fully represented by the set $C$. Together with Theorem~\ref{representation_tuple} this leads to a complete characterisation of $[C,I]^{ec}$ in Lemma~\ref{orbitrep}.
In its proof, we also need the following.
\begin{figure}
\centering
\begin{tikzpicture}
\draw (0,0) -- (7,0) node[right]{$\mathbb{Q}$};
\draw (0,2) -- (7,2) node[right]{$\mathbb{Q}$};

\draw (1,1.8) -- (1,2.2) node[above]{$C(1)$};
\draw (2,1.8) -- (2,2.2) node[above]{$C(2)$};
\draw (6,1.8) -- (6,2.2) node[above]{$C(3)$};

\draw (2,0.2) -- (2,-0.2) node[below]{$C'(1)$};
\draw (4,0.2) -- (4,-0.2) node[below]{$C'(2)$};
\draw (5,0.2) -- (5,-0.2) node[below]{$C'(3)$};

\draw[-{Latex}] (1,2) -- (2,0);
\draw[-{Latex}] (2,2) -- (4,0);
\draw[-{Latex}] (6,2) -- (5,0);

\draw[dotted] (0,1) -- (0.5,0);
\draw[dotted] (0,2) -- (1,0);
\draw[dotted] (0.5,2) -- (1.5,0);

\draw[dotted] (4/3,2) -- (8/3,0);
\draw[dotted] (5/3,2) -- (10/3,0);

\draw[dotted] (3,2) -- (4.25,0);
\draw[dotted] (4,2) -- (4.5,0);
\draw[dotted] (5,2) -- (4.75,0);

\draw[dotted] (6.5,2) -- (5.5,0);
\draw[dotted] (7,2) -- (6,0);
\draw[dotted] (7,1) -- (6.5,0);
\end{tikzpicture}
\caption{Visualisation of $\pi$ from Lemma~\ref{lem:support_homogeneity}}\label{fig:mappingvis}
\end{figure}
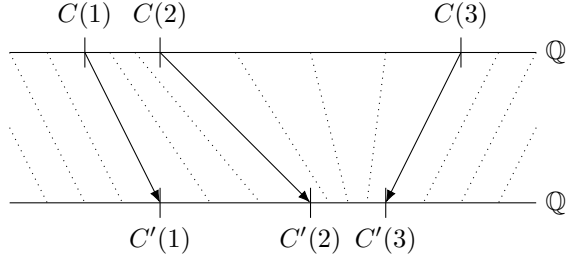

\begin{lemma}[Homogeneity]\label{lem:support_homogeneity}
For any two finite $C \subseteq\mathbb{Q}$, $C' \subseteq\mathbb{Q}$, if $|C| = |C'|$, then there is a $\pi\in G$ such that $C\pi = C'$.
\end{lemma}
\begin{proof}
This is shown through construction of $\pi$. Number the elements of $C$ from smallest to largest, such that $C(1)$ is the smallest element and $C(n)$ the largest. Do the same for $C'$. We define $\pi$ such that $C(i)\pi = C'(i)$, interpolating in between (see Figure~\ref{fig:mappingvis} for a visualisation):
\[ \pi(x) = \begin{cases}
              x - C(1) + C'(1)                                     &\text{if } x < C(1) \\
              (x-C(i))\tfrac{C'(i+1)-C'(i)}{C(i+1)-C(i)} + C'(i)   &\text{if } C(i) \leq x < C(i+1) \\
              x-C(n) + C'(n)                                       &\text{if } C(n) \leq x
\end{cases} \]
Note that since $\frac{C'(i)-C'(i-1)}{C(i)-C(i-1)} > 0$ for any $1 < i \le n$, $\pi$ is monotone. Furthermore, its inverse is given by swapping $C$ and $C'$.
\[ \pi^{-1}(x) = \begin{cases}
              x - C'(1) + C(1)                                     &\text{if } x < 'C(1) \\
              (x-C'(i))\tfrac{C(i+1)-C(i)}{C'(i+1)-C'(i)} + C(i)   &\text{if } C'(i) \leq x < C'(i+1) \\
              x-C'(n) + C(n)                                       &\text{if } C'(n) \leq x
\end{cases} \]
Hence $\pi$ is a monotone bijection, and we conclude $\pi\in G$.
\end{proof}

\begin{lemma}\label{orbitrep}
Given a finite subset $C \subset \mathbb{Q}$, we have $[C,I]^{ec} \cong \mathcal{P}_{|C|}(\mathbb{Q})$.
\end{lemma}
\begin{proof}
From Lemma~\ref{lem:support_homogeneity} 
it follows that $\mathcal{P}_{|C|}(\mathbb{Q})$ consists of a single orbit. Given this, in combination with the fact that $C\in\mathcal{P}_{|C|}(\mathbb{Q})$, Theorem~\ref{representation_tuple} gives a subgroup $S\le G|_C$ such that $\mathcal{P}_{|C|}(\mathbb{Q}) \cong [C,S]^{ec}$. Since $S\le G|_C$, Lemma~\ref{group_trivial} implies $S = I$. This proves that $[C,I]^{ec} \cong \mathcal{P}_{|C|}(\mathbb{Q})$.
\end{proof}

By Theorem~\ref{representation_tuple} and the above lemmas, we can represent an orbit by a single integer $n$, the size of the least support of its elements.
\begin{corollary}\label{cor:single-orbit}
	Let $X$ be an orbit-finite nominal set. Then
	$X \cong \mathcal{P}_{\dim(X)}(\mathbb{Q})$.
\end{corollary}
\begin{proof}
By Theorem~\ref{representation_tuple}, we get $C$ and $S \le G|_C$ such that $X \cong [C,S]^{ec}$. 
By Lemma~\ref{group_trivial}, $S = I$, and by Lemma~\ref{orbitrep}
we get 
\[ X \cong [C,I]^{ec} \cong \mathcal{P}_{|C|}(\mathbb{Q}) \,. \]
But $|C| = \dim(X)$, since $C$ is the least support of some element $x \in X$. 
\end{proof}
This naturally extends to (orbit-finite) nominal sets with multiple orbits
by using a multiset of natural numbers, representing the size of the least support of each of the orbits. 
These multisets are formalised here as functions $f \colon \mathbb{N} \rightarrow \mathbb{N}$. 

\begin{definition}\label{def:present-nom}
Given a function $f \colon \mathbb{N}\rightarrow\mathbb{N}$, we define a nominal set $\extendset{f}$ by
\begin{align*}
\extendset{f} = \bigcup_{\substack{n\in\mathbb{N}\\1\le i\le f(n)}} \{i\}\times \mathcal{P}_n(\mathbb{Q}).
\end{align*}
\end{definition}

\begin{proposition}\label{thm:pres-nom}
For every orbit-finite nominal set $X$, there is a unique function $f \colon \mathbb{N} \to \mathbb{N}$ such that $X \cong \extendset{f}$ and the set $\{n \mid f(n) \neq 0\}$ is finite.
\end{proposition}
\begin{proof}
We start by proving the existence. For this, grade $X$ by the dimension of its elements, defining $X_i = \{x\in X \mid \dim(x) = i\}$. Now split each $X_i$ up into its $k_i$ orbits $O_{i,j}$, such that 
\[ X_i = \bigcup\limits_{1\le j \le k_i} O_{i,j} \, . \]
By Corollary~\ref{cor:single-orbit}, we have $O_{i,j}\cong \{j\}\times \mathcal{P}_i(\mathbb{Q})$
for each orbit $O_{i,j}$. 

Define $f\colon\mathbb{N}\rightarrow\mathbb{N}$ such that $f(i) = k_i$. Then 
$\{n \mid f(n) \neq 0\}$ is finite, since $X$ is orbit-finite.
Writing this out gives 
\begin{align*}
\extendset{f} = \bigcup_{\substack{i\in\mathbb{N}\\1\le j\le f(i)}} \{j\}\times \mathcal{P}_i(\mathbb{Q})
\cong \bigcup_{\substack{i\in\mathbb{N}\\1\le j\le k_i}} O_{i,j}
= X.
\end{align*}

Next, we need to show that $f$ is unique. Suppose $g\colon\mathbb{N}\rightarrow\mathbb{N}$ also represents $X$, e.g. $X\cong\extendset{g}$. Then it follows that $\extendset{f}\cong\extendset{g}$. Let $h\colon \extendset{f}\rightarrow \extendset{g}$ be the isomorphism. Grade $\extendset{f}$ and $\extendset{g}$, letting $\extendset{f}_i = \{x\in \extendset{f} \mid \dim(x) = i\}$, and similarly for $\extendset{g}_i$. Since $h$ is an isomorphism, we have for any $x\in \extendset{f}$ that $\dim(h(x)) = \dim(x)$, implying $h(\extendset{f}_i) = \extendset{g}_i$. Furthermore, the fact that $h$ is an isomorphism gives $\Nsize(h(\extendset{f}_i)) = \Nsize(\extendset{f}_i)$. Using $\Nsize(\extendset{f}_i) = f(i)$, we find that $f(i) = \Nsize(\extendset{f}_i) = \Nsize(h(\extendset{f}_i)) = \Nsize(\extendset{g}_i) = g(i)$. Hence $f = g$, proving that $f$ is unique.
\end{proof}

\begin{example}\label{ex:qtimesq}
Consider the set $\mathbb{Q}\times\mathbb{Q}$. The elements $(a,b)$ split in three orbits, one for $a<b$, one for $a=b$ and one for $a>b$. These have dimension 2, 1 and 2 respectively, so the set $\mathbb{Q}\times\mathbb{Q}$ is represented by the multiset $\{1,2,2\}$.
\end{example}

\begin{remark}
The representation in terms of a function $f \colon \mathbb{N}\rightarrow\mathbb{N}$
enforces that there are only finitely many orbits of any given dimension.
The first part of the above proposition 
generalises to arbitrary nominal sets by replacing the codomain
of $f$ by the class of all sets and adapting Definition~\ref{def:present-nom} accordingly. 
However, the choice of function $f$ is no longer unique in that case---this is avoided in the finite
case by using natural numbers instead of finite sets.
\end{remark}

\subsection{Equivariant maps}
\label{sec:eq-maps}

We show how to represent equivariant maps, using two basic properties.
Let $f \colon X \to Y$ be an equivariant map.
The first property is that the direct image of an orbit (in $X$) is again an orbit (in $Y$),
that is to say, $f$ is defined `orbit-wise'.
Second, equivariant maps cannot introduce new elements in the support (but they can drop them).
More precisely:
\begin{lemma}\label{lem:eqiorbit}
Let $f \colon X \to Y$ be an equivariant map, and $O\subseteq X$ a single orbit.
The direct image $f(O) = \{f(x) \mid x \in O\}$ is a single-orbit nominal set.
\end{lemma}
\begin{proof}
Let $y$ and $y'$ both be elements of $f(O)$. To show that $f(O)$ is single-orbit, we need to construct a $\pi \in G$ such that $y\pi = y'$. By definition of $f(O)$, there exist $x\in O$, $x'\in O$ such that $f(x) = y$ and $f(x') = y'$. Since $O$ is single-orbit, there exists a $\pi \in G$ such that $x\pi = x'$. As $f$ is an equivariant function, we find $y\pi = f(x)\pi = f(x\pi) = f(x') = y'$. This proves that $f(O)$ is single-orbit.
\end{proof}

\begin{lemma}\label{eqimap_contained}
Let $f \colon X \to Y$ be an equivariant map between two nominal sets $X$ and $Y$. Let $x\in X$ and let $C$ be a support of $x$.
Then $C$ supports $f(x)$.
\end{lemma}
\begin{proof}
Let $\pi\in G$ be such that $\forall c \in C, c\pi = c$. Then since $C$ is the support of $x$, $x\pi = x$. But then also $f(x)\pi = f(x\pi) = f(x) = f(x)$. Hence $C$ is a support of $f(x)$. Then, by definition, the least support of $f(x)$ is contained in $C$.
\end{proof}

Hence, equivariant maps are fully determined by associating two pieces of information for each orbit in the domain:
the orbit on which it is mapped and a string denoting which elements of the least support of the input are preserved.
These ingredients are formalised in the first part of the following definition. 
The second part describes how these ingredients
define an equivariant function.
Proposition~\ref{eqimap_rep} below
then states that every equivariant function can be
described in this way. 

\begin{definition}\label{def:eqmap}
Let $H = \{(I_1, F_1, O_1), \ldots, (I_n, F_n, O_n)\}$ be a finite set of tuples where the $I_i$'s are disjoint single-orbit nominal sets, the $O_i$'s are single-orbit nominal sets with $\dim(O_i) \leq \dim(I_i)$ and the $F_i$'s are bit strings of length $\dim(I_i)$ with exactly $\dim(O_i)$ ones. 

Given a set $H$ as above, we define
$f_H \colon \bigcup I_i \to \bigcup O_i$ as the unique equivariant function such that, given $x \in I_i$ with least support $C$, $f_H(x)$ is the unique element of $O_i$ with support $\{C(j) \mid F_i(j) = 1\}$, where $F_i(j)$ is the $j$-th bit of $F_i$ and $C(j)$ is the $j$-th smallest element of $C$.
\end{definition}

\begin{proposition}\label{eqimap_rep}
For every equivariant map $f \colon X\rightarrow Y$ between orbit-finite nominal sets $X$ and $Y$ there exists a unique set $H$ 
as in Definition~\ref{def:eqmap} such that $f = f_H$.
\end{proposition}
\begin{proof}
We start with showing existence by construction. Split $X$ into its constituent orbits, call them $I_1$ through $I_n$. For each of these, select an element $e_i \in I_i$. Let $O_i$ be the orbit of $f(e_i)$. By Lemma~\ref{lem:eqiorbit}, $f(I_i) = O_i$. For each $e_i$, let $C_i$ be the least support of $e_i$ and $C'_i$ the least support of $f(e_i)$. Let $F_i$ be the string with $F_i(j) = 1$ if $C_i(j) \in C'_i$, and $F_i(j)=0$ otherwise. Let $H = \{(I_i, F_i, O_i) \mid i \in \{1,\ldots, n\}\}$. By construction, $f_H(e_i)$ is the unique element of $O_i$ with support $C'_i \cap C_i$. By Lemma~\ref{eqimap_contained}, $C'_i\cap C_i = C'_i$, implying $f_H(e_i) = f(e_i)$. Since both are equivariant functions with the same domain, we have $f(x) = f_H(x)$ for all $x \in X$. Hence $f = f_H$.

To show that $H$ is unique, consider an $H'$ such that $f=f_{H'}$. As a consequence, we have $f_H = f_{H'}$. From the definition of orbits it follows immediately that the split of $X$ into $I_i$ is unique up to the choice of indices, so that we can label the tuples $(I_i', F_i', O_i')$ in $H'$ such that $I_i' = I_i$. It then follows that $O_i = f_H(I_i) = f_{H'}(I'_i) = O'_i$. To show $F_i = F_i'$, consider an $x\in I_i$. Let $C$ denote the least support of $x$, and $C_f$ the least support of $f(x)$. By definition of $f_H$ and $f_{H'}$, it follows that $\{C(j) \mid F_i(j) = 1\} = C_x = \{C(j) \mid F_i'(j)$. But this is only possible if $F_i = F_i'$, and hence $H=H'$.
\end{proof}

\begin{example}
Consider the function $\min\colon\mathcal{P}_3(\mathbb{Q})\rightarrow\mathbb{Q}$ which returns the smallest element of a 3-element set. Note that both $\mathcal{P}_3(\mathbb{Q})$ and $\mathbb{Q}$ are single orbits. Since for the orbit $\mathcal{P}_3(\mathbb{Q})$ we only keep the smallest element of the support, we can thus represent the function $\min$ with $\{(\mathcal{P}_3(\mathbb{Q}), 100, \mathbb{Q})\}$.
\end{example}
\begin{example}
Consider the (right) projection $\pi_2 \colon \mathbb{Q} \times \mathbb{Q} \rightarrow \mathbb{Q}$.
Recall from Example~\ref{ex:qtimesq} that the set $\mathbb{Q} \times \mathbb{Q}$ has three orbits 
$Q_1 = \{(a,b) \mid a < b\}$, $Q_2 = \{(a,b) \mid a=b\}$ and $Q_3 = \{(a,b) \mid a>b\}$.
The function $\pi_2$ is represented by $\{(Q_1, 01, \mathbb{Q}),(Q_2, 1, \mathbb{Q}), (Q_3, 10, \mathbb{Q})$. 
\end{example}

\subsection{Products}\label{subsec:products}

The product $X \times Y$ of two nominal sets is again a nominal set
and hence it can be represented itself in terms of the dimension of each of its orbits as shown in Section~\ref{sec:orbits-and-nominal-set}. However, this approach has some disadvantages.

\begin{example}\label{ex:product}
We start by showing that the orbit structure of products can be non-trivial.
Consider the product of $X=\mathbb{Q}$ and the set ${Y=\{(a,b) \in \Q^2 \mid a < b\}}$.
This product consists of five orbits, more than one might naively expect from the fact that both sets are single-orbit:
\[
\begin{array}{ll}
\{(a,(b,c)) \mid a,b,c\in\mathbb{Q}, a < b < c\}, \quad & \{(a,(a,b)) \mid a,b\in\mathbb{Q}, a < b\}, \\
\{(b,(a,c)) \mid a,b,c\in\mathbb{Q}, a < b < c\}, \quad & \{(b,(a,b)) \mid a,b\in\mathbb{Q}, a < b\}, \\
\{(c,(a,b)) \mid a,b,c\in\mathbb{Q}, a < b < c\}.
\end{array}
\]
\end{example}

We find that this product is represented by the multiset $\{2,2,3,3,3\}$. Unfortunately, this is not sufficient to accurately describe the product as it abstracts
away from the relation between its elements with those in $X$ and $Y$.
In particular, it is not possible to reconstruct the projection maps from such a representation. 

The essence of our representation of products
is that each orbit $O$ in the product $X \times Y$ is described entirely by the dimension of $O$ together with 
the two (equivariant) projections
$\pi_1 \colon O \to X$ and $\pi_2 \colon O \to Y$.
This combination of the orbit and the two projection maps can already be represented using Propositions~\ref{thm:pres-nom} and \ref{eqimap_rep}. However, as we will see, a combined representation for this has several advantages. For discussing such a representation, let us first introduce what it means for tuples of a set and two functions to be isomorphic:

\begin{definition}
Given nominal sets $X, Y, Z_1$ and $Z_2$, and equivariant functions $l_1 \colon Z_1 \to X$, $r_1 \colon Z_1 \to Y$, $l_2 \colon Z_2 \to X$ and $r_2 \colon Z_2 \to Y$, we define $(Z_1, l_1, r_1) \cong (Z_2, l_2, r_2)$ if there exists an isomorphism $h \colon Z_1 \to Z_2$ such that $l_1 = l_2 \circ h$ and $r_1 = r_2 \circ h$.
\end{definition}

Our goal is to have a representation that, for each orbit $O$, produces a tuple $(A, f_1, f_2)$ isomorphic to the tuple $(O, \pi_1, \pi_2)$. The next lemma gives a characterisation that can be used to simplify such a representation.

\begin{lemma}\label{lm:supp-product}
Let $X$ and $Y$ be nominal sets and $(x, y) \in X \times Y$.
If $C$, $C_x$ and $C_y$ are the least supports of $(x, y)$, $x$ and $y$ respectively, then $C = C_x \cup C_y$.
\end{lemma}
\begin{proof}
Let $\pi\in G$ be a group element such that $\forall c \in C_x \cup C_y, c\pi = c$.
Then
$(x,y)\pi = (x \pi, y \pi) = (x,y)$
since $C_x$ and $C_y$ are supports of $x$ and $y$ respectively. 
Hence $C_x \cup C_y$ is a support of $x$, and since $C$ is the least support of $x$, $C\subseteq C_x\cup C_y$.

Now suppose that $C$ is strictly smaller than $C_x\cup C_y$. Then there is an element $c \in C_x\cup C_y$ with $c \notin C$. Without loss of generality we can assume $c \in C_x$. 
The set  $(C_x\cup C_y)\setminus\{c\}$ is \emph{not} a support of $x$, 
since the least support $C_x$ of $x$ is not contained in $(C_x\cup C_y)\setminus\{c\}$.
Hence, there is some $\pi \in G$ such that $\forall c' \in (C_x\cup C_y)\setminus\{c\}, c' \pi = c'$, but $x\pi \ne x$. 
For this $\pi$, we have $(x, y)\pi = (x\pi, y\pi) \ne (x, y)$. 
However, $(C_x\cup C_y)\setminus\{c\}$ is a support of $(x,y)$, since $C \subseteq (C_x\cup C_y)\setminus\{c\}$. Hence
$(x,y) \pi = (x,y)$, yielding a contradiction.
\end{proof}

With Proposition~\ref{eqimap_rep} we represent the maps $\pi_1$ and $\pi_2$ 
by tuples $(O, F_1, O_1)$ and $(O, F_2, O_2)$ respectively.
Using Lemma~\ref{lm:supp-product} and the definitions of $F_1$ and $F_2$, we see that at least one of $F_1(i)$ and $F_2(i)$ equals $1$ for each $i$.

We can thus combine the strings $F_1$ and $F_2$ into a single string $P \in \{L, R, B\}^{*}$ as follows.
We set $P(i) = L$ when only $F_1(i)$ is $1$, $P(i) = R$ when only $F_2(i)$ is $1$, and $P(i) = B$ when both are 
$1$.
The string $P$ fully describes the strings $F_1$ and $F_2$. This process for constructing the string $P$ gives it two useful properties:
\begin{itemize}
	\item The number of $L$s and $B$s in the string $P$ equals the dimension of $O_1$. 
	\item The number of $R$s and $B$s in the string $P$ equals the dimension of $O_2$. 
\end{itemize}
We will call strings $P$ with the above two properties \emph{valid} (with respect to $O_1,O_2$).

Thus, to describe a single orbit of the product $X \times Y$, a valid string $P$ together with the images of $\pi_1$ and $\pi_2$ is sufficient. This is stated more precisely in Proposition~\ref{orbstring}.

\begin{definition}\label{def:orbstringdef}
Let $O_1\subseteq X$, $O_2\subseteq Y$ be single-orbit sets, and 
let $P \in\{L,R,B\}^{*}$ be a valid string with respect to $O_1,O_2$. Define
\begin{align*}
\extendproduct{(P,O_1,O_2)} = (\mathcal{P}_{|P|}(\mathbb{Q}), f_{H_1}, f_{H_2}),
\end{align*}
where $H_i = \{(\mathcal{P}_{|P|}(\mathbb{Q}), F_i, O_i)\}$
and the string $F_1$ is defined as the string $P$ with $L$s and $B$s replaced by $1$s and $R$s by $0$s.
The string $F_2$ is similarly defined with the roles of $L$ and $R$ swapped.
\end{definition}
This construction generates orbits of $X \times Y$:
\begin{lemma}
	Let $(P,O_1,O_2)$ be a tuple as in Definition~\ref{def:orbstringdef}.
	Then we have $\extendproduct{(P,O_1,O_2)} \cong (O, \pi_1, \pi_2)$ 
	for some orbit $O \subseteq X \times Y$. 
\end{lemma}
\begin{proof}
Let $(O', f, g) = \extendproduct{(P, O_1, O_2)}$. By construction, we find $f(O') \subseteq X$ and $g(O') \subseteq Y$. 
Denote by $\langle f, g \rangle \colon O' \rightarrow X \times Y$ the pairing,
i.e., $\langle f, g \rangle(x) = (f(x),g(x))$.
By Lemma~\ref{lem:eqiorbit}, since $O'$ is single-orbit, so is $\langle f, g \rangle(O')$. The latter
is an orbit of $X \times Y$. 

We now show that $\langle f, g \rangle$ is an isomorphism. 
First, by construction of $f$ and $g$, we find that if $C$ is the least support of $x\in O'$, then $C$ is also the least support of $(f(x), g(x))$, since every element in the support of $x$ is in at least one of the least supports of $f(x)$ and $g(x)$, and by Lemma~\ref{lm:supp-product}, the least support of $(f(x), g(x))$ is the union of the least supports of $f(x)$ and $g(x)$. This implies that the elements of both $O'$ and $\langle f, g \rangle(O')$ have the same support size. Since both $O'$ and $\langle f, g \rangle(O')$ are single-orbit, this makes $\langle f, g \rangle$ a bijection. Hence,
\[
\extendproduct{(P, O_1, O_2)} = (O', f, g) \cong (\langle f, g \rangle(O'), \pi_1|_{\langle f, g \rangle(O')}, \pi_2|_{\langle f, g \rangle(O')}) \,. \qedhere
\]
\end{proof}
The following result shows that every orbit of $X \times Y$ arises in this way, up to isomorphism.

\begin{proposition}\label{orbstring}
For every orbit $O \subseteq X \times Y$
there is a unique tuple 
$(P, O_1, O_2)$ such that $O_1\subseteq X$, $O_2\subseteq Y$ are orbits, $P$ is a valid string
and
\[ \extendproduct{(P,O_1,O_2)} \cong (O, \pi_1|_O, \pi_2|_O) \,. \]
\end{proposition}

\begin{proof}
Let us start by constructing such a tuple $(P, O_1, O_2)$ for a given orbit $O\subseteq X\times Y$. Since $O$ is an orbit, Lemma~\ref{eqimap_rep} provides two tuples $(O, F_1, O_1)$ and $(O, F_2, O_2)$ with $O_1$ an orbit of $X$ and $O_2$ an orbit of $Y$ such that $\pi_1|_O = f_{\{(O,F_1, O_1)\}}$ and $\pi_2|_O = f_{\{(O,F_2,O_2)\}}$. 
Now construct $P$ as the sequence of length $\dim(O)$ as follows: 
\[
P(i) = 
\begin{cases}
L & \text{ if } F_1(i) = 1 \text{ and } F_2(i) = 0 \\
R & \text{ if } F_1(i) = 0 \text{ and } F_2(i) = 1 \\
B & \text{ if } F_1(i) = F_2(i) = 1
\end{cases}
\]
This covers all cases, since Lemma~\ref{lm:supp-product} guarantees that it will 
never be the case that the $i$-th letters of $F_1$ and $F_2$ are both $0$.

We first show that $P$ is valid. To see this, observe that by Definition~\ref{def:eqmap}, each $1$ in the string $F_1$ corresponds to a unique element in the least support of an element of $O_1$. Hence $\dim(O_1) = |F_1|$. By definition of $F_1$ we find $\dim(O_1) = |F_1| = |P|_L+|P|_B$. A similar line of reasoning, replacing $L$ with $R$, shows $\dim(O_2) = |F_2| = |P|_R+|P|_B$.

To show that the correspondence is bijective, 
we first show that, given an orbit $O \subseteq X \times Y$, 
if  $(P,O_1,O_2)$ is the corresponding triple then 
\[
(O, \pi_1|_O, \pi_2|_O) \cong \extendproduct{(P,O_1,O_2)} \,.
\]
Let $n=\dim(O)$. By Corollary~\ref{cor:single-orbit}, we have $\mathcal{P}_n(\mathbb{Q}) \cong O$. 
Let $g \colon O \stackrel{\cong}{\rightarrow} \mathcal{P}_n(\mathbb{Q})$ be the isomorphism between them. 
Then $f_{\{(O,F_1, O_1)\}} = f_{\{(\mathcal{P}_{|F_1|}(\mathbb{Q}), F_1, O_1)\}} \circ g$ and $f_{\{(O,F_2,O_2)\}} = f_{\{(\mathcal{P}_{|F_2|}(\mathbb{Q}),F_2,O_2)\}} \circ g$.
As a consequence, we have
\[ (O,\pi_1|_O,\pi_2|_O) \cong (\mathcal{P}_{|F_1|}(\mathbb{Q}), f_{\{(\mathcal{P}_{|F_1|}(\mathbb{Q}), F_1, O_1)\}},f_{\{(\mathcal{P}_{|F_2|}(\mathbb{Q}),F_2,O_2)\}}) \,. \]
By Definition~\ref{def:orbstringdef} and the above construction of $P$, 
the right-hand side of the above equation equals $\extendproduct{(P, O_1, O_2)}$.
Hence, $(O, \pi_1|_O, \pi_2|_O) \cong \extendproduct{(P,O_1,O_2)}$.

Finally, for uniqueness, consider two tuples $(P, O_1, O_2)$ and $(P', O_1', O_2')$, such that
$\extendproduct{(P, O_1, O_2)} \cong \extendproduct{(P', O_1', O_2')}$.
We let $F_1$ and $F_2$ denote the strings from Definition~\ref{def:orbstringdef} for $\extendproduct{(P, O_1, O_2)}$, 
and similarly $F_1'$ and $F_2'$ the strings for $\extendproduct{(P', O_1', O_2')}$. 

Since $\extendproduct{(P, O_1, O_2)} \cong \extendproduct{(P', O_1', O_2')}$ we have $\mathcal{P}_{|P|}(\mathbb{Q})\cong \mathcal{P}_{|P'|}(\mathbb{Q})$, hence $|P| = |P'|$. 
Furthermore, by the isomorphism, for any $x\in \mathcal{P}_{|P|}(\mathbb{Q})$, there exists an $x'\in\mathcal{P}_{|P'|}(\mathbb{Q})$ such that 
\begin{equation}\label{eq:xinproof}
\begin{array}{rcl}
f_{\{(\mathcal{P}_{|P|}(\mathbb{Q}), F_1, O_1)\}}(x) &=& f_{\{(\mathcal{P}_{|P'|}(\mathbb{Q}), F_1', O_1')\}}(x')\,, \text{ and } \\ f_{\{(\mathcal{P}_{|P|}(\mathbb{Q}), F_2, O_2)\}}(x) &=& f_{\{(\mathcal{P}_{|P'|}(\mathbb{Q}), F_2', O_2')\}}(x') \,.
\end{array}
\end{equation}
Since $O_1$, $O_2$, $O_1'$ and $O_2'$ single-orbit, this implies $O_1 = O_1'$ and $O_2 = O_2'$. 

Moreover, by Lemma~\ref{lm:supp-product}, the least support of any element $x\in \mathcal{P}_{|P|}(\mathbb{Q})$
equals the least support of 
$(f_{\{(\mathcal{P}_{|P|}(\mathbb{Q}), F_1, O_1)\}}(x), f_{\{(\mathcal{P}_{|P|}(\mathbb{Q}), F_2, O_2)\}}(x))$.
But if we choose $x'$ corresponding to $x$ as in the previous paragraph, then 
by~\eqref{eq:xinproof} we obtain that the least support of $x$ equals the least support of $x'$. 
Hence $x=x'$. 
But this implies that 
$f_{\{(\mathcal{P}_{|P|}(\mathbb{Q}), F_1, O_1)\}} = f_{\{(\mathcal{P}_{|P'|}(\mathbb{Q}), F_1', O_1)\}}$
and $f_{\{(\mathcal{P}_{|P|}(\mathbb{Q}), F_2, O_2)\}} = f_{\{(\mathcal{P}_{|P'|}(\mathbb{Q}), F_2', O_2)\}}$,
hence also $F_1 = F_1'$ and $F_2 = F_2'$. But $F_1 = F_1'$ and $F_2 = F_2'$ can hold only if $P$ and $P'$ are equal. From this, we conclude that $(P, O_1, O_2) = (P', O_1', O_2')$.
\end{proof}

From the above proposition it follows that we can generate the product $X\times Y$ simply by enumerating all valid strings $P$ for all pairs of orbits $(O_1, O_2)$ of $X$ and $Y$.
Given this, we can calculate the multiset representation of a product from the multiset representations of both factors.

\begin{theorem}\label{thm:prod-combinatorics}
For $X \cong \extendset{f}$ and $Y \cong \extendset{g}$ we have $X \times Y \cong \extendset{h}$, where
\begin{align*}
h(n) = \sum_{\substack{0 \le i, j \le n \\ i+j \ge n}} f(i)g(j){{n}\choose{j}}{{j}\choose{n-i}}.
\end{align*}
\end{theorem}
\begin{proof}
Every string $P\in\{L,R,B\}^{*}$ of length $n$ with $|P|_L = n-j$, $|P|_R = n-i$ and $|P|_B = i+j-n$ satisfies the requirements of Lemma~\ref{orbstring}, and hence describes a unique orbit for every pair of orbits $O_1$ and $O_2$ where the least support of the elements of $O_1$ has size $i$, and the least support of elements of $O_2$ have size $j$. Combinatorics tells us that there are ${{n}\choose{j}}{{j}\choose{n-i}}$ such strings. Summing over all $i\ge 0$, $j\ge 0$ such that $i+j-n$, $n-j$ and $n-i$ are positive, and multiplying with the number of orbits of the required size gives the result.
\end{proof}

\begin{example}\label{ex:prodstrings}
To illustrate some aspects of the above representation, let us use it to calculate the product of Example~\ref{ex:product}.
First, we observe that both $\Q$ and $S=\{(a,b) \in \Q^2 \mid a < b\}$ consist of a single orbit.
Hence any orbit of the product corresponds to a triple $(P, \Q, S)$, where the string $P$ satisfies $|P|_L+|P|_B = \dim(\Q) = 1$ and $|P|_R+|P|_B = \dim(S) = 2$.
We can now find the orbits of the product $\Q \times S$ by enumerating all strings satisfying these equations.
This yields:
\begin{itemize}
\item LRR, corresponding to the orbit $\{(a,(b,c)) \mid a,b,c\in\Q, a < b < c\}$,
\item RLR, corresponding to the orbit $\{(b,(a,c)) \mid a,b,c\in\Q, a < b < c\}$,
\item RRL, corresponding to the orbit $\{(c,(a,b)) \mid a,b,c\in\Q, a < b < c\}$,
\item RB, corresponding to the orbit $\{(b,(a,b)) \mid a,b\in\Q, a < b\}$, and
\item BR, corresponding to the orbit $\{(a,(a,b)) \mid a,b\in\Q, a < b\}$.
\end{itemize}
Each product string fully describes the corresponding orbit.
To illustrate this, consider the string BR.
The corresponding bit strings for the projection functions are $F_1 = 10$ and $F_2 = 11$.
From the lengths of the string we conclude that the dimension of the orbit is $2$. The string $F_1$ further tells us that the left element of the tuple consists only of the smallest element of the support. The string $F_2$ indicates that the right element of the tuple is constructed from both elements of the support. Combining this, we find that the orbit is $\{(a,(a,b)) \mid a,b\in\Q, a < b\}$.
\end{example}

\subsection{Summary}
We summarise our concrete representation in the following table.
Propositions~\ref{thm:pres-nom}, \ref{eqimap_rep} and \ref{orbstring} correspond to the three rows in the table.

\begin{table}[h]
\centering
\renewcommand{\arraystretch}{1.2} 
\begin{tabular}{l|p{5.7cm}}
	\emph{Object} & \emph{Representation}  \\
	\hline
	Single orbit $O$                & Natural number $n = \dim(O)$ \\
	Nominal set $X = \bigcup_i O_i$ & Multiset of these numbers  \\
	\hline
	Map from single orbit $f \colon O \to Y$ & The orbit $f(O)$ and a bit string $F$  \\
	Equivariant map $f \colon X \to Y$       & Set of tuples $(O, F, f(O))$, one for each orbit \\
	\hline
	Orbit in a product $O \subseteq X \times Y$ & The corresponding orbits of $X$ and $Y$, and a string $P$ relating their supports  \\
	Product $X \times Y$                        & Set of tuples $(P, O_X, O_Y)$, one for each orbit  \\
\end{tabular}
\renewcommand{\arraystretch}{1.0}
\caption{Overview of representation.}\label{tab:overview}
\end{table}

Notice that in the case of maps and products, the orbits are inductively represented using the concrete representation.
As a base case we can represent single orbits by their dimension.

\subsection{Comparison to \emph{Families of symmetries} by Ciancia et al.}
\label{sec:rel-work-ciancia}

\newcommand*{\cat}[1]{\textbf{#1}} 
\newcommand*{\catC}{\mathcal{C}}

The representation we have derived so far is the same as in the work of Ciancia et al~\cite{CianciaKM10}.
In this section we elaborate on this equivalence.
Their particular result is an encoding of nominal sets (in a general sense) as objects in the category $\cat{Fam}(\cat{Sym}(\catC)^{\textrm{op}})$, for a suitably picked $\catC$.

In our setting of the total order symmetry, we have to take $\catC$ to be the small category of finite sets with monotone injections.
Concretely, the objects of $\catC$ are the sets $\{1, \ldots, n\}$ for each $n \in \N$ and the maps of $\catC$ are injective and order-preserving functions $f \colon \{1, \ldots, n\} \to \{1, \ldots, m\}$.
(Note that an object $\{1, \ldots, n\}$ can just as well be encoded by the number $n$.)
By Lemma~\ref{group_trivial}, we note that $\cat{Sym}(\catC)$ is isomorphic to $\catC$.
And so their representation $\cat{Fam}(\cat{Sym}(\catC)^{\textrm{op}}) \cong \cat{Fam}(\catC^{\textrm{op}})$ also consists of multisets of natural numbers, as does ours (see Table~\ref{tab:overview}).
The maps are encoded per member in the family, together with local (backwards) maps.
We have specialised this encoding of maps as follows.
A monotone injection $f \colon n \to m$ can be recovered from its image $\im(f) \subseteq m$.
Such a subset can be encoded as a bit string, which is what we use.

They also provide a concrete description of the categorical product, in terms of the \emph{multi-coproduct} (see~\cite[Section~5]{CianciaKM10}).
We have shown that an element in the multi-coproduct can be represented concretely by a certain string.

Finally, we mention \cite[Theorem~4.12]{CianciaKM10} which states that there is an equivalence $\cat{Fam}(\cat{Sym}(\catC)^{\textrm{op}}) \to \cat{Set}^{\catC}_\diamond$.
Since we have the same representation of nominal sets, this suggests that the category of nominal sets over $(\Q, <)$ is equivalent to the category $\cat{Set}^{\catC}_\diamond$ of (wide) pullback preserving functors from $\cat{C}$ to $\cat{Set}$.
However, in this paper we have not defined the categorical underpinnings to further elaborate on this equivalence.

\section{\cplusplus{} Implementation of \ONS{}}
\label{sec:impl}
\lstset{basicstyle=\ttfamily}

The ideas outlined above have been implemented in the \cplusplus{} library \ONS{}.\footnote{\ONS{} can be found at \url{https://github.com/davidv1992/ONS}}
The library can represent orbit-finite nominal sets and their products, (disjoint) unions and maps. A full technical description of what it can do, and how to use it, is given in the documentation included with \ONS{}.

Let us start here by showing an example program to calculate the product of the sets $\{(a,b)\mid a<b\}$ and $\mathbb{Q}$
(see Example~\ref{ex:prodstrings}). The program below calculates this product, and then prints one element for each orbit of the result.

{\small
\begin{lstlisting}[language=C++]
nomset<rational> A = nomset_rationals();
nomset<pair<rational, rational>> B({rational(1),rational(2)});

auto AtimesB = nomset_product(A, B);   // compute the product
for (auto orbit : AtimesB)
	cout << orbit.getElement() << " ";
\end{lstlisting}
}

In the first line, we create a nominal set $A$, and initialise it with the built-in set $\mathbb{Q}$. The type of such a nominal set variable is \lstinline+nomset<T>+, where $T$ is the type of the elements. In case of $A$ this is the \lstinline+rational+ type.

In the second line, we create the set $B$ containing the elements of $\{(a,b)\mid a<b\}$. To do this, we instruct the constructor of $B$ to create the minimal nominal set containing the element $(1,2)$. This creates the nominal set with the orbit of $(1,2)$, which is exactly the set $\{(a,b)\mid a<b\}$.

Having created these sets, it then computes the product using the function \lstinline+nomset_product+. This returns the product, of type \lstinline+nomset<pair<A,B>>+, where \lstinline+A+ and \lstinline+B+ are the types of the elements of $A$ and $B$ respectively. In our case, this means that the result is a nominal set containing elements of type \lstinline+pair<rational,pair<rational,rational>>+.

Finally, we loop over the orbits of the result, stored in the variable \lstinline+AtimesB+, with \lstinline+for (auto orbit : AtimesB)+. This returns an \lstinline+orbit+ object for each orbit in the nominal set \lstinline+AtimesB+. These objects describe the properties of the individual orbits of a nominal set. We use it here to get an element (with \lstinline+.getElement()+), which is printed through standard out.

Running this code gives the following output
(`\texttt{\small /1}' signifies the denominator):
{\small
\begin{lstlisting}
(1/1,(2/1,3/1)) (1/1,(1/1,2/1)) (2/1,(1/1,3/1))
(2/1,(1/1,2/1)) (3/1,(1/1,2/1))
\end{lstlisting}
}

We see here a list of five elements, each corresponding to a single orbit of the product $\mathbb{Q}\times\{(a,b)\mid a < b\}$.

\subsection{Core functionality}

The main implementation of nominal sets and equivariant functions in the \ONS{} library is split up in three main concepts: 
\begin{enumerate}
	\item \lstinline+orbit<T>+, representing orbits;
	\item \lstinline+nomset<T>+, representing nominal sets, containing a number of \lstinline+orbit<T>+ objects;
	\item \lstinline+eqimap+, equivariant functions.
\end{enumerate}
The \lstinline+orbit<T>+ objects contain a complete description of a single orbit of elements of type $T$. They allow querying of basic properties such as the size of the least support (with \lstinline+.supportSize+), checking whether an element is a member of the orbit (with \lstinline+.isElement+) and extraction of sample  elements (with \lstinline+.getElement+).

Next, \lstinline+nomset<T>+ is used to represent entire sets. Its functionality includes checking whether an element or other set is contained in the set (with \lstinline+.contains+), iterating over the orbits (see above) and querying for the size of the set (with \lstinline+.size+).

For working with nominal sets, \ONS{} also provides implementations of common set operations. Examples of these are set union (with \lstinline+nomset_union+), intersection (with \lstinline+nomset_intersect+) and set products (with \lstinline+nomset_product+).

The \ONS{} library also implements support for filtering (with \lstinline+nomset_filter+) and mapping (with \lstinline+nomset_map+) of nominal sets. These take an (equivariant) function as argument, which can either be given as an equivariant function object \lstinline+eqimap+ or as a \cplusplus{} function or function object, as long as the resulting behaviour when invoked is equivariant.

Finally, objects of type \lstinline+eqimap+ can be used to represent dynamically generated equivariant functions. They implement an evaluation, allowing the application of the function to concrete argument values. They also contain several functions for querying properties of the function represented (such as whether elements are in its domain, with \lstinline+.inDomain+), and manipulating the function represented (such as extending the mapping, with \lstinline+.add+).

\subsection{Another example}

Let us now consider a slightly more complicated piece of code.
It refines a relation $R$ (called \lstinline+previousPartition+) on the states $Q$ of an automaton, using a transition function $f\colon Q\times A\rightarrow A$ (called \lstinline+transitionFunction+), over an alphabet $A$ (called \lstinline+alphabet+). 
It returns as a result the set 
\[
\{(q,q')\in R \mid \forall a\in A\colon (f(q,a), f(q',a)) \in R\} \,.
\]
The code is as follows.
{\small
\begin{lstlisting}[language=C++]
template <typename Q, typename A>
nomset<pair<Q,Q>> refineRelation(
	nomset<A> alphabet,
	nomset<pair<Q,Q>> previousRelation,
	eqimap<pair<Q,A>, Q> transitionFunction) {

	// calculate R x A
	nomset<pair<pair<Q,Q>,A>> transitions =
		nomset_product(previousRelation, alphabet);

	// Find those where (f(q,a), f(q',a)) not in R
	nomset<pair<pair<Q,Q>,A>> invalid =
		nomset_filter(transitions, [&](pair<pair<Q,Q>,A> input) {
			Q state1 = input.first.first;
			Q state2 = input.first.second;
			A letter = input.second;
			
			Q result1 = transitionFunction({state1, letter});
			Q result2 = transitionFunction({state2, letter});
			
			return !previousRelation.contains({result1, result2});
		});

	// Strip away alphabet
	nomset<pair<Q,Q>> toRemove =
		nomset_map(invalid, [](pair<pair<Q,Q>,A> input) {
			return input.first;
		});

	// Calculate result
	return nomset_minus(previousRelation, toRemove);
}
\end{lstlisting}
}
This result is computed in four steps. First, it calculates the set of all transition pairs it still needs to consider (transitions), using \lstinline+nomset_product+ to calculate the product $R\times A$.

Next, it uses \lstinline+nomset_filter+ to select those elements $((q,q'),a)$ for which $(f(q,a), f(q',a)) \notin R$. Note that the function which calculates this is specified as a plain \cplusplus{} lambda, whose behaviour is equivariant.

For modifying the relation, only the first part of $((q,q'),a)$ is relevant, so in the third step we use \lstinline+nomset_map+ to project out the alphabet letters of the witnesses in the set \lstinline+invalid+.

This leaves the code with a set of pairs $(q,q')$ which need to be removed from $R$ to produce the final result. This is calculated using the set minus operation \lstinline+nomset_minus+, resulting in the refined partition.

\section{Haskell Implementation of \ONS{}}
\label{sec:haskell}

We have implemented a similar library in Haskell, called \ONShs.\footnote{Available at \url{https://github.com/Jaxan/ons-hs}}
This showcases the generality of the theoretical characterisation in Section~\ref{sec:total-order}.

At the core, there is the type class \hask{Nominal}.
\begin{haskell}
class Nominal a where
  type Orbit a :: *
  toOrbit :: a -> Orbit a
  getElement :: Orbit a -> Support -> a
  support :: a -> Support
  index :: Proxy a -> Orbit a -> Int
\end{haskell}
It provides to basic functionality, \hask{toOrbit} and \hask{getElement}, to convert between elements (of type \hask{a}) and orbits of elements (of type \hask{Orbit a}).
The functions \hask{support} and \hask{index} are utility functions, returning the (least) support of an element, and the dimension of an orbit.

Instances are defined for basic data types such as the type of rational numbers, \hask{Atom}.
Other instances can be derived for any algebraic data structure (following Table~\ref{tab:overview}).
For example, the data type for the states of the automaton in Figure~\ref{fig:lint-minimal} can be defined as:
\begin{haskell}
data State = Q0 | Q1 Atom | Q2 Atom Atom | Q3 Atom Atom | Q4
  deriving (Eq, ...)
  deriving Nominal via Generic State
\end{haskell}

Deriving instances with generics make it easy for the user to use the library.
One can also easily define trivial instances, where the group action is defined as the identity function.
This is used for the \hask{EquivariantSet} data structure.
This data type provides an interface to (infinite) nominal sets and the usual set constructions are defined.
Some of these functions are shown below (for brevity, we have omitted the type class context from the type signatures).
\begin{haskell}
data EquivariantSet a = ...
  deriving Nominal via Trivial (EquivariantSet a)

map :: (a -> b) -> EquivariantSet a -> EquivariantSet b
filter :: (a -> Bool) -> EquivariantSet a -> EquivariantSet a
product :: EquivariantSet a -> EquivariantSet b -> EquivariantSet (a, b)
rationals :: EquivariantSet Atom
\end{haskell}

Function arguments (in, e.g., \hask{map} and \hask{filter}) are required to be equivariant.
Finally, a data type for equivariant maps, \hask{EquivariantMap} is provided with the expected functions for a map data structure.

With these functions, we can define all the states of the automaton in Figure~\ref{fig:lint-minimal} and the accepting states can easily be filtered out.
\begin{haskell}
states = fromList [Q0, Q4]  <>  map Q1 rationals
  <> map (uncurry Q2) (product rationals rationals)
  <> map (uncurry Q3) (product rationals rationals)
acceptingStates = filter accept states where
  accept (Q2 a b) = a < b
  accept _ = False
\end{haskell}

\section{Complexity of set operations}\label{sec:complexity}

Since our implementations are directly based on representing orbits, it is possible to derive concrete complexities for the set operations. To simplify such an analysis, we make the following assumptions on operations on orbits:
\begin{itemize}
\item The comparison of two orbits takes $O(1)$.
\item Constructing an orbit from an element takes $O(1)$.
\item Checking whether an element is in an orbit takes $O(1)$.
\end{itemize}
These assumptions are justified as each of these operations takes time proportional to the size of the representation of an individual orbit, which in practice is small and approximately constant. For instance, the orbit $\mathcal{P}_n(\Q)$ is represented by just the integer $n$ and its type.

Furthermore, two of the operations considered make use of an external function. Since these can be implemented in a variety of ways, and the time complexity of actually invoking these functions is highly dependent both on what it calculates, and the specific way it is implemented in the program, we will here simply consider invocations of these functions to take $O(1)$ time.

For the notation in the following statement, recall that $\Nsize(X)$ denotes the number of orbits of $X$, and $\dim(X)$ the maximal size of the least support of its elements.

\begin{theorem}\label{thmcomplexity}
If nominal sets are implemented with a tree-based set structure (as in \ONS{}), the complexity of the following set operations is as follows:

\begin{center}
\begin{tabular}{r|l}
Operation & Complexity\\
\hline
Test $x \in X$ & $O(\log \Nsize(X))$\\
Test $X\subseteq Y$ & $O(\min(\Nsize(X)+\Nsize(Y), \Nsize(X)\log \Nsize(Y)))$\\
Calculate $X \cup Y$ & $O(\Nsize(X)+\Nsize(Y))$\\
Calculate $X \cap Y$ & $O(\Nsize(X)+\Nsize(Y))$\\
Calculate $\{x \in X \mid p(x)\}$ & $O(\Nsize(X))$\\
Calculate $\{f(x) \mid x \in X\}$ & $O(\Nsize(X)\log \Nsize(X))$\\
Calculate $X\times Y$ & $O(\Nsize(X\times Y)) \,\subseteq\, O(3^{\dim(X)+\dim(Y)}\Nsize(X)\Nsize(Y))$\\
\end{tabular}
\end{center}

The functions $p\colon X\rightarrow 2$ and $f\colon X\rightarrow Y$ are user defined, and assumed to take $O(1)$ time per invocation.
\end{theorem}
\begin{proof}
Each of the statements will be proven individually:

\emph{Membership.}
To decide $x \in X$, we first construct the orbit containing $x$, which is done in constant time.
Then we use a logarithmic lookup to decide whether this orbit is in our set data structure.
Hence membership checking is $O(\log(\Nsize(X)))$.

\emph{Inclusion.}
We can check $X \subseteq Y$ in $O(N(X)\log(N(Y)))$ by repeated membership.
It is also possible to do a simultaneous in-order traversal of both sets, which takes $O(N(X)+N(Y))$ time.
The implementation uses a cutoff on the size of $X$ relative to $Y$ to deal with this, giving a time complexity of $O(min(N(X)+N(Y), N(X)\log(N(Y))))$.

\emph{Union and Intersection.}
A simultaneous traversal through both sets $X$ and $Y$ can be used to compute their union and intersection.
This gives a complexity of $O(N(X)+N(Y))$ for intersections and unions.

\emph{Filtering.}
Filtering a nominal set $X$ using some equivariant predicate $p$ can be done in linear time, as the results are obtained in order, giving a complexity of $O(N(X))$, assuming the time complexity of the function $p$ to be constant.

\emph{Mapping.}
Mapping is a harder than filtering, as the outputs of $f$ may be out of order.
Hence, for a tree-based implementation of sets, a sorting step is needed (or equivalently, iterated insertions), which gives a complexity to $O(N(X)\log(N(X)))$, again assuming the time complexity of the function $f$ to be constant.

\emph{Products.}
Calculating the product of two nominal sets is the most complicated construction.
For each pair of orbits in the original sets $X$ and $Y$, all product orbits need to be generated.
Each product orbit itself is constructed in constant time.
By generating these orbits in-order, the resulting set takes $O(\Nsize(X\times Y))$ time to construct.

We can also give an explicit upper bound for the number of orbits in terms of the input.
Recall that orbits in a product are represented by strings of length at most $\dim(X)+\dim(Y)$.
(If the string is shorter, we pad it with one of the symbols.)
Since there are three symbols ($L, R$ and $B$), the product of $X$ and $Y$ will have at most $3^{\dim(X)+\dim(Y)}\Nsize(X)\Nsize(Y)$ orbits.
It follows that taking products has time complexity of $O(3^{\dim(X)+\dim(Y)}\Nsize(X)\Nsize(Y))$.
\end{proof}

Using the above complexity results on individual operations, we can derive the complexity of algorithms using nominal sets. We will demonstrate this here using Moore's algorithm. Recall from Section~\ref{sec:minimisation}:

\begin{algorithm}
\caption{Moore's minimisation algorithm for nominal DFAs}
\begin{algorithmic}[1]
\Require{ Nominal automaton $(S,A,F,\delta)$.}
\State{$i \leftarrow 0$, ${\equiv_{-1}} \leftarrow S\times S$, ${\equiv_{0}} \leftarrow F\times F \cup (S\backslash F)\times (S\backslash F)$}
\While{${\equiv_i} \,\ne\, {\equiv_{i-1}}$}
\State{${\equiv_{i+1}} \leftarrow \{(q_1, q_2) \mid q_1 \equiv_i q_2 \,\wedge\, \forall a\in A, \delta(q_1,a) \equiv_i \delta(q_2,a) \}$}
\State{$i \leftarrow i+1$}
\EndWhile
\State{$E\leftarrow S/_{\equiv_i}$}
\State{$F_E \leftarrow \{e\in E \mid \forall s\in e, s\in F\}$}
\State{Let $\delta_E$ be the map such that, if $s\in e$ and $\delta(s,a)\in e'$, then $\delta_E(e,a) = e'$.}
\State{\Return{$(E,A,F_E,\delta_E)$.}}
\end{algorithmic}
\end{algorithm}

\begin{theorem}\label{thm:moore}
The runtime complexity of Moore's algorithm on nominal deterministic automata is $O(3^{5k} k \Nsize(S)^3 \Nsize(A))$, where $k=\dim(S\cup A)$.
\end{theorem}

\begin{proof}
	This is shown by counting operations, using the complexity results of set operations stated in Theorem~\ref{thmcomplexity}.
	We first focus on the while loop on lines 2 through 5.
	The runtime of an iteration of the loop is determined by line 3, as this is the most expensive step.
	Since the dimensions of $S$ and $A$ are at most $k$, computing $S \times S \times A$ takes $O(\Nsize(S)^2\Nsize(A)3^{5k})$.
	Filtering $S \times S$ using that then takes $O(\Nsize(S)^23^{2k})$.
	The time to compute $S \times S \times A$ dominates, hence each iteration of the loop takes $O(\Nsize(S)^2\Nsize(A)3^{5k})$.
	
	Next, we need to count the number of iterations of the loop.
	Each iteration of the loop gives rise to a new partition, which is a refinement of the previous partition.
	Furthermore, every partition generated is equivariant. Note that this implies that each refinement of the partition does at least one of two things:
	distinguish between two orbits of $S$ previously in the same element(s) of the partition, or distinguish between two members of the same orbit previously in the same element of the partition.
	The first can happen only $\Nsize(S)-1$ times, as after that there are no more orbits lumped together.
	The second can only happen $\dim(S)$ times per orbit, because each such a distinction between elements is based on splitting on the value of one of the elements of the support.
	Hence, after $\dim(S)$ times on a single orbit, all elements of the support are used up. Combining this, the longest chain of partitions of $S$ has length at most $O(k\Nsize(S))$.
	
	Since each partition generated in the loop is unique, the loop cannot run for more iterations than the length of the longest chain of partitions on $S$. It follows that there are at most $O(k\Nsize(S))$ iterations of the loop, giving the loop a complexity of $O(k\Nsize(S)^3\Nsize(A)3^{5k})$
	
	The remaining operations outside the loop have a lower complexity than that of the loop, hence the complexity of Moore's minimisation algorithm for a nominal automaton is $O(k\Nsize(S)^3\Nsize(A)3^{5k})$.
\end{proof}

The above theorem shows in particular that minimisation of nominal automata is fixed-parameter tractable (FPT) with the dimension as fixed parameter.
The complexity of Algorithm~\ref{alg:moore} for nominal automata is very similar to the $O(|S|^3 |A|)$ bound given by a naive implementation of Moore's algorithm for ordinary DFAs.
This suggest that it is possible to further optimise an implementation with similar techniques used for ordinary automata.

\section{Evaluation}
\label{sec:evaluation}

This section presents an experimental evaluation of \ONS{} and \ONShs{}, comparing them against
existing tools in two tasks related to nominal automata: learning and minimisation.

\subsection{The Tools: \ONS{}, \ONShs{}, \NLambda{} and \LOIS{}}
In order to evaluate our library \ONSboth{}, we compare it against two existing libraries for computing with nominal sets, \NLambda{} \cite{EPTCS207.3} and \LOIS{} \cite{kopczynski1716,kopczynski2017}.
We briefly describe how these tools work and what the differences are with \ONS{}.

Both \NLambda{} and \LOIS{} work symbolically.
Nominal sets are represented with set-builder expressions: values with variables and a first-order formula describing those values.
For example, the orbit $\{ \{a, b, c\} \mid a,b,c \in \Q, a < b < c \}$ is represented simply as:
\[ \{ \text{\hask{fromList} } [x_0, x_1, x_2] \mid x_0 < x_1 \wedge x_1 < x_2 \}. \]
(Here \hask{fromList} takes a list and constructs a set.)
In \ONS{}, orbits are represented more compactly: in this case, the integer $3$ suffices.
On the other hand, a set such as $\Q^3$ has a compact representation in \NLambda{} and \LOIS:
\[ \{ (x_0, x_1, x_2) \mid \top \} \]
and a larger representation in \ONS{}, consisting of 13 strings (see Section~\ref{subsec:products}).

Since \NLambda{} and \LOIS{} use formulas, many set operations are expressed by manipulating these formulas.
One of the crucial operations is determining whether a set is empty, which can be resolved by checking satisfiability of the formula.
To this end, these libraries use an SMT solver (by default, both libraries use Z3 \cite{MouraB08}).
Consequently, the runtime will depend on the size of these formulas, and both libraries have routines to simplify formulas.

\subsection{Benchmarks}

We evaluate the scalability of each library by implementing the automata minimisation algorithm and learning algorithm
discussed in Section~\ref{sec:aut}.
These are then tested on the following three sets of automata.

\paragraph{Structured automata}
We define the following automata.
\begin{itemize}
	\item[FIFO($n$)] Automata accepting valid traces of a finite FIFO data structure of size $n$.
	The alphabet is defined by two orbits: $\{\text{Put}(a) \mid a \in \Q \}$ and $\{\text{Get}(a) \mid a \in \Q \}$.
	\item[$ww(n)$] Automata accepting the language of words of the form $ww$, where $w \in \Q^{n}$.
	\item[$\Lmax$]
	The language $\Lmax$ where the last symbol is the maximum of previous symbols (Example~\ref{ex:minimal-automaton}).
	\item[$\Lint$] The language accepting a series of nested intervals (Example~\ref{ex:int-lang}).
\end{itemize}

The first two classes of structured automata are used as test cases in \cite{moerman2017learning}.
These two classes are also equivariant w.r.t.\ the equality symmetry.
The structured automata can be encoded directly in symbolic form in \NLambda{} or \LOIS{}.
In \ONS{}, this structure is lost and the algorithms operate purely on orbits.
Where applicable, the automata listed above were generated using the same code as used in \cite{moerman2017learning}, ported to the other libraries as needed.

The automaton accepting the FIFO language is not minimal.
It is based on the purely functional queue which has two lists of data values: one for pushing, one for popping.
If the list for popping is empty, the list for pushing is reversed and moved to the list for popping~\cite{Okasaki99}.
The use of two lists is redundant and hence the automaton is not minimal.

\paragraph{Orbit-wise random automata}
Besides structured automata, we generate orbit-wise random automata as follows.
The input alphabet is always $\mathbb{Q}$ and the number of orbits and dimension $k$ of the state space $S$ are fixed.
For each orbit in the set of states, its dimension is chosen uniformly at random between $0$ and $k$, inclusive.
Each orbit has a probability $\frac{1}{2}$ of consisting of accepting states.

To generate the transition function $\delta$, we enumerate the orbits of $S \times \Q$ and choose a target state uniformly from the orbits $S$ with small enough dimension.
The bit string indicating which part of the support is preserved is then sampled uniformly from all valid strings. We will denote these automata as rand$_{\Nsize(S),k}$.
The choices made here are arbitrary and only provide basic automata.
We note that the automata are generated orbit-wise and this may favour our tool.

\paragraph{Random automata with formulae}
The two classes above (orbit-wise random and structured) are very different in nature.
The random ones are defined orbit-wise (which is an advantage for \ONS{}),
whereas the structured ones hardly use the values in an interesting way.
To provide a middle-ground, we also generate random automata which use formulas on transitions.

For these automata, the state space is constructed from multiple copies of $\mathbb{Q}^n$.
We will refer to these copies as locations and we will refer to the number of locations in the state space as the size of the automaton.
We fix the size of the automaton, and then for each location we sample its dimension $n$ uniformly from $[0, k]$, where $k$ is a chosen constant.
This creates the set of states. Every location has a probability $\frac{1}{2}$ of being accepting. 
Note that $\mathbb{Q}^n$ consists of more than one orbit if $n>1$, and consequently, the number of orbits in the state of the resulting automata can vary.
The alphabet used is always the set $\mathbb{Q}$.

To generate the transition function, we generate a formula for each of the locations.
This is done by creating a tree, where every node is one operation.
The tree starts with an empty root, and is then expanded by repeatedly selecting one of the empty nodes and replacing it with either a logical operation (`and' or `or') with two new empty nodes, or by a literal, i.e., a comparison between two variables (either $<$, $=$, or $>$).
Both options occur with equal chance.
In the latter case, the variables are drawn from either the state values or the input value.
This process is repeated until either there are no more empty nodes, or the limit on the number of logical operators is reached, at which point the remaining empty nodes are filled with literals.
Finally, for each node in the tree we randomly invert the output or not.

These formulas are then used to create two edges: one for when the formula is true, and one for when the formula is false.
The target state is specified by randomly choosing a location, and randomly selecting which elements of the original state and input are kept in the target location (we allow for duplicate values).

\paragraph{Properties of the random automata}

Our main motivation for using the above described random automata is the unavailability of a good source of nominal automata used in practical applications. However, this makes it difficult to judge whether or not our random automata are representative for actual performance in real test cases.
We will show some properties of the generated automata so the reader can make their own judgement.
In particular, we will focus primarily on the degree to which the size of the automata is reduced during minimisation.

\begin{figure}
\centering
\begin{subfigure}[b]{0.33\textwidth}
\includegraphics[width=\textwidth]{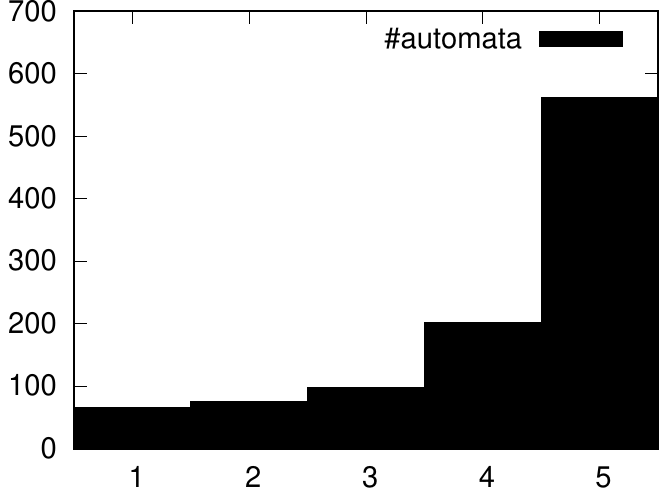}
\caption{$N(S) = 5, k = 1$}
\end{subfigure}~
\begin{subfigure}[b]{0.33\textwidth}
\includegraphics[width=\textwidth]{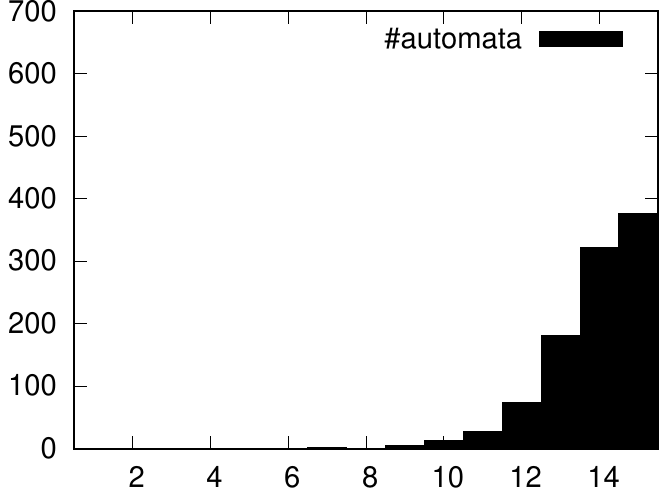}
\caption{$N(S) = 15, k = 1$}
\end{subfigure}~
\begin{subfigure}[b]{0.33\textwidth}
\includegraphics[width=\textwidth]{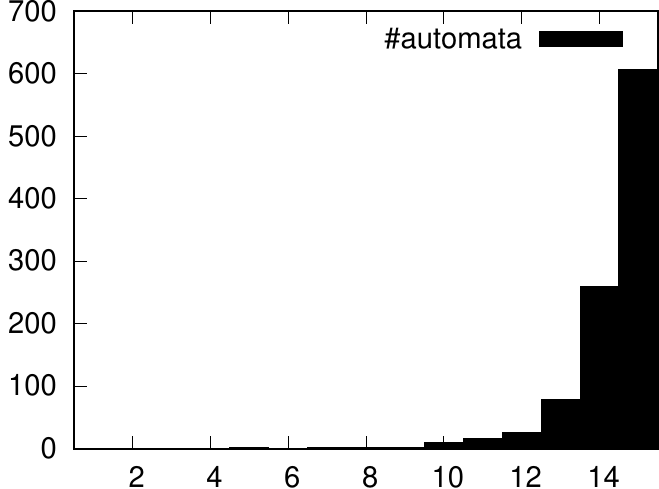}
\caption{$N(S) = 15, k = 3$}
\end{subfigure}
\caption{Histogram of the number of orbits after minimisation for orbit-wise random automata. Each figure shows $1000$ automata. The number of orbits before minimisation is $N(S)$ and the dimension is $k$.}
\label{fig:rnd1}
\end{figure}

For the orbit-wise generated automata, this is shown in Figure~\ref{fig:rnd1}.
It can be clearly seen that the vast majority of the generated automata is either minimal, or very close to being minimal.

\begin{figure}
\centering
\begin{subfigure}[b]{0.45\textwidth}
\includegraphics[width=\textwidth]{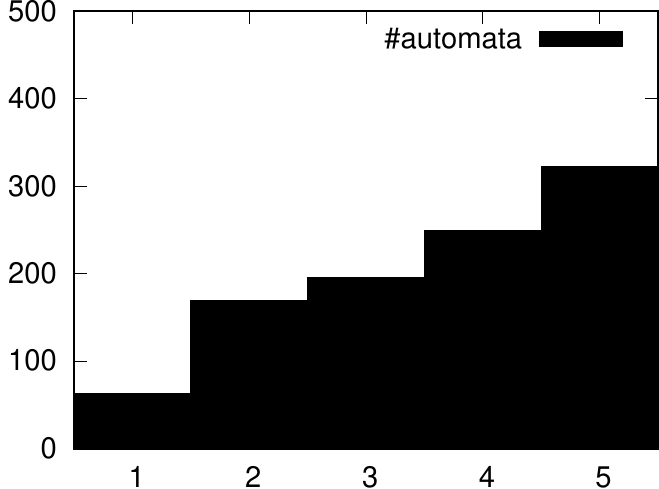}
\caption{$5$ locations, $k = 1$}
\end{subfigure}~
\begin{subfigure}[b]{0.45\textwidth}
\includegraphics[width=\textwidth]{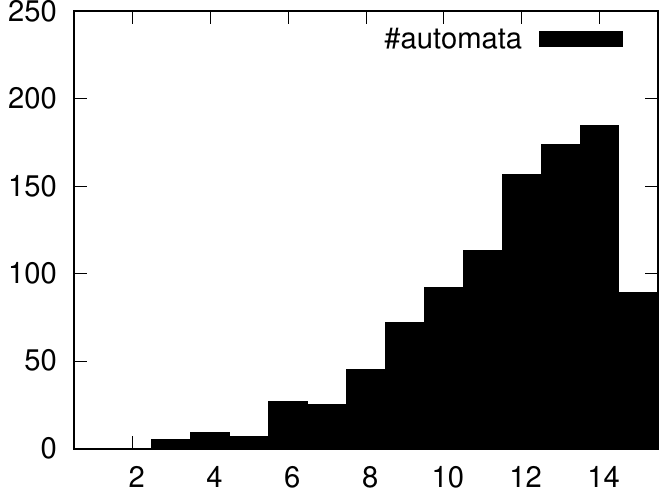}
\caption{$15$ locations, $k = 1$}
\end{subfigure}
\caption{Number of state orbits after minimisation for formula automata of size $5$ and $15$, of dimension $1$. The figure shows $1000$ automata. Note that because the dimension is $1$, the number of orbits before minimisation is always $5$ or $15$ respectively.}
\label{fig:for1}
\end{figure}
\begin{figure}
\centering
\begin{subfigure}[b]{0.45\textwidth}
\includegraphics[width=\textwidth]{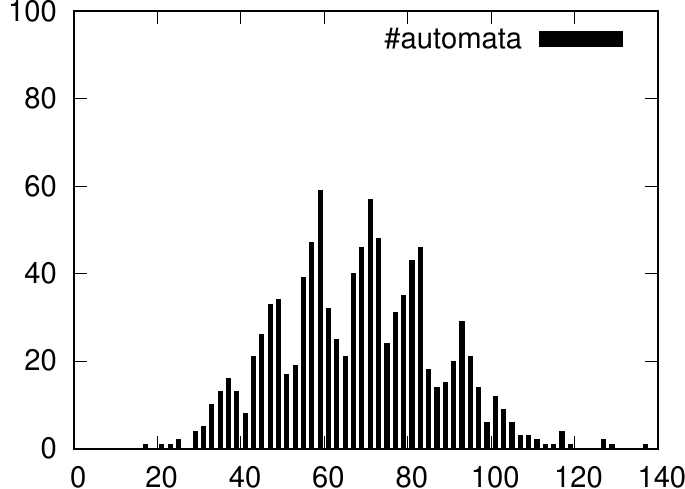}
\end{subfigure}~
\begin{subfigure}[b]{0.45\textwidth}
\includegraphics[width=\textwidth]{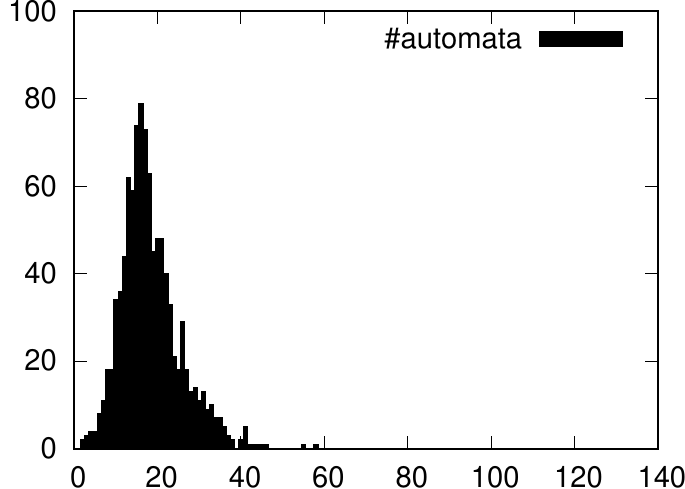}
\end{subfigure}
\caption{Number of state orbits before (left) and after (right) minimisation for formula automata of size $15$ and dimension $3$. Each figure shows $1000$ automata. The striped pattern on the left is a consequence of the fact that the used settings only generate automata with an odd number of orbits.}\label{fig:for3}
\end{figure}

In contrast, the data for the formula automata (see Figures~\ref{fig:for1} and~\ref{fig:for3}) shows a much broader distribution, generating automata that use significantly more orbits than needed for the recognised languages.

Given this, we think that our test cases are varied enough to give a decent representation of the performance of the various libraries.

\subsection{Minimisation Results}
For \NLambda{} and \LOIS{} we used the implementations of Moore's minimisation algorithm from the original papers~\cite{EPTCS207.3,kopczynski1716,kopczynski2017}.
For each of the libraries, we wrote routines to read in an automaton from a file and,
for the structured test cases, to generate the requested automaton.
For \ONS{}, all automata were read from file.
The output of these programs was manually checked to see if the minimisation was performed correctly.

The results in Table~\ref{minimize_results} show a clear advantage for \ONS{} for random automata.
The library is capable of running all supplied test cases in less than one second.
This in contrast to both \LOIS{} and \NLambda{}, which take more than $2$ hours on the largest random automata.

\begin{table}[]
\centering
\begin{tabular}{l|rrr|rrrr}
Model        & \multicolumn{1}{c}{$N$} & \multicolumn{1}{c}{$L$} & \multicolumn{1}{c|}{$\dim$} & \multicolumn{1}{c}{\ONS{} (s)} & \multicolumn{1}{c}{\ONShs{} (s)} & \multicolumn{1}{c}{\NLambda{} (s)} & \multicolumn{1}{c}{\LOIS{} (s)} \\ \hline
Random       & 5     &       & $\leq$ 1  & 0.00     & 0.00     & 0.05    & 0.71     \\
Random       & 10    &       & $\leq$ 1  & 0.00     & 0.00     & 0.93    & 12.31    \\
Random       & 10    &       & $\leq$ 2  & 0.01     & 0.00     & 26.60   & $>$ 2h   \\
Random       & 15    &       & $\leq$ 1  & 0.00     & 0.00     & 5.01    & 41.31    \\
Random       & 15    &       & $\leq$ 2  & 0.01     & 0.00     & 61.98   & $>$ 2h   \\
Random       & 15    &       & $\leq$ 3  & 0.04     & 0.02     & 418.22  & $>$ 2h   \\ \hline
Formula      &       & 5     & $\leq$ 2  & 0.00     & 0.01     & 0.25    &          \\
Formula      &       & 10    & $\leq$ 2  & 0.01     & 0.02     & 1.52    &          \\
Formula      &       & 10    & $\leq$ 3  & 0.62     & 0.57     & 1.52    &          \\
Formula      &       & 25    & $\leq$ 2  & 0.06     & 0.05     & 48.81   &          \\
Formula      &       & 25    & $\leq$ 3  & 2.89     & 1.58     & 108.46  &          \\
Formula      &       & 25    & $\leq$ 5  & $>$ 2h   & 2176.83  & 255.39  &          \\ \hline
FIFO($1$)    & 4     & 4     & 1         & 0.02     & 0.00     & 0.01    & 0.03     \\
FIFO($2$)    & 13    & 7     & 2         & 0.01     & 0.01     & 1.03    & 0.24     \\
FIFO($3$)    & 65    & 11    & 3         & 0.35     & 0.70     & 9.32    & 2.44     \\
FIFO($4$)    & 440   & 16    & 4         & 37.60    & 39.77    & 68.44   & 15.33    \\
FIFO($5$)    & 3686  & 22    & 5         & $>$ 2h   & 3027.54  & 382.33  & 71.59    \\ \hline
$ww(1)$       & 4     & 4     & 1         & 0.02     & 0.00     & 0.01    & 0.03     \\
$ww(2)$       & 8     & 6     & 2         & 0.00     & 0.00     & 0.12    & 0.03     \\
$ww(3)$       & 24    & 8     & 3         & 0.16     & 0.17     & 0.75    & 0.16     \\
$ww(4)$       & 112   & 10    & 4         & 23.71    & 30.51    & 2.85    & 0.60     \\
$ww(5)$       & 728   & 12    & 5         & 5880.04  & $>$ 2h   & 9.27    & 1.83     \\ \hline
$\Lmax$      & 5     &       & 2         & 0.02     & 0.00     & 1.08    & 0.06     \\
$\Lint$      & 5     &       & 2         & 0.01     & 0.00     & 1.71    & 0.04
\end{tabular}
\caption{Running times for Algorithm~\ref{alg:moore}.
$\Nsize$ ($\dim$) is the size (resp. dimension) of the input.
The column $L$ denotes the number of locations, if the automaton can be expressed symbolically.
The first two sets of automata (`Random' and `Formula') are the randomly generated automata.
Each rows consists of 10 automata and we report the average runtime.
If one of the runs times out, the cell indicates $>$ 2h.}
\label{minimize_results}
\end{table}

The results for structured automata show a clear effect of the extra structure.
Both \NLambda{} and \LOIS{} remain capable of minimising the automata in reasonable amounts of time for larger sizes.
In contrast, \ONS{} benefits little from the extra structure. Despite this, it remains viable: even for the larger cases it falls behind significantly only for the largest FIFO automaton and the two largest $ww$ automata.

The random automata with formulae show a mixed bag (as expected).
(We have not implemented this for LOIS.)
We note that when the number of locations grow, \NLambda{} becomes rather slow.
Nonetheless, \NLambda{} catches up when increasing the dimension.
All the results show that \ONSboth{} is faster in lower dimensions, even with a high number of orbits, but that \NLambda{} and \LOIS{} can handle higher dimensions.

The libraries \ONS{} and \ONShs{} are very comparable in terms of scalability.
However, we note that \ONShs{} is sometimes faster, we expect this is due to the lazy nature of Haskell:
When iterating through a set (especially those formed by products), breaking early means that the remainder of the set does not need to be constructed.

\subsection{Learning Results}
\label{sec:learning-results}
Both implementations in \NLambda{} and \ONS{} are direct implementations of the pseudocode for \nLStar{} with
no further optimisations.
The authors of \LOIS{} implemented \nLStar{} in their library as well.\footnote{Can be found on \url{github.com/eryxcc/lois/blob/master/tests/learning.cpp}}
They reported similar performance as the implementation in \NLambda{} (private communication).
Hence we focus our comparison on \NLambda{} and \ONS{}.
We use the variant of \nLStar{} where counterexamples are added as columns instead of prefixes.

The implementation in \NLambda{} has the benefit that it can work with different symmetries.
Indeed, the structured examples, FIFO and $ww$, are equivariant w.r.t.\ the equality symmetry as well as the total order symmetry.
For that reason, we run the \NLambda{} implementation using both the equality symmetry and the total order symmetry on those languages.
For the languages $\Lmax$, $\Lint$ and the random automata, we can only use the total order symmetry.

To run the \nLStar{} algorithm, we implement an external oracle for the membership queries.
This is akin to the application of learning black box systems~\cite{Vaandrager17}.
For equivalence queries, we constructed counterexamples by hand.
All implementations receive the same counterexamples.
We measure CPU time instead of real time, so that we do not account for the external oracle.

The results in Table~\ref{learning_results} show an advantage for \ONS{} for random automata.
Additionally, we report the number of membership queries,
which can vary for each implementation as some steps in the algorithm depend on the internal ordering of set data structures.

We have not benchmarked the learning algorithm with the random automata with formulae.
The reason is that the logical information, the formulae, are not even given to the learning algorithm, since the learning algorithm can only query the language.

In contrast to the case of minimisation,
the results suggest that \NLambda{} cannot exploit the logical structure of FIFO$(n)$, $\Lmax$ and $\Lint$ as it is not provided a priori.
For $ww(2)$ we inspected the output on \NLambda{} and saw that it learned some logical structure.
For example, it outputs $\{ (a, b) \mid a \neq b \}$ as a single object instead of two orbits $\{ (a, b) \mid a < b \}$ and $\{ (a, b) \mid b < a \}$.
This may explain why \NLambda{} is still competitive.
For languages which are equivariant for the equality symmetry,
the \NLambda{} implementation using the equality symmetry
can learn with much fewer queries.
This is expected as the automata themselves have fewer orbits.
It is interesting to see that these languages can be learned more efficiently by choosing the right symmetry.

\begin{table}
\makebox[\textwidth][c]{{\scriptsize
\begin{tabular}{l|rr|rr|rr|rr|rr}
        & & & \multicolumn{2}{l|}{\ONS{}} & \multicolumn{2}{l|}{\ONShs{}} & \multicolumn{2}{l|}{\NLambda{}$^{\mathit{ord}}$} & \multicolumn{2}{l}{\NLambda{}$^{\mathit{eq}}$} \\
Model     & $N$ & $\dim$ & time (s) & MQs     & time (s)& MQs     & time (s) & MQs   & time (s) & MQs   \\ \hline
Random    & 4   & 1      & 127.47   & 2321    & 6.94    & 1915    & 2391.08  & 1243  &          &       \\
Random    & 5   & 1      & 0.12     & 404     & 0.08    & 404     & 2433.77  & 435   &          &       \\
Random    & 3   & 0      & 0.86     & 499     & 0.14    & 470     & 1818.97  & 422   &          &       \\
Random    & 5   & 1      & $>$ 1h   &         & 192.18  & 6870    & $>$ 1h   &       &          &       \\
Random    & 4   & 1      & 0.08     & 387     & 0.06    & 387     & 2097.43  & 387   &          &       \\ \hline
FIFO($1$) & 3   & 1      & 0.04     & 119     & 0.01    & 119     & 3.17     & 119   & 1.76     & 51    \\
FIFO($2$) & 6   & 2      & 1.73     & 2655    & 0.55    & 2655    & 391.89   & 3818  & 40.00    & 434   \\
FIFO($3$) & 19  & 3      & 2793.93  & 298400  & 451.67  & 302868  & $>$ 1h   &       & 2047.32  & 8151  \\ \hline
$ww(1)$    & 4   & 1      & 0.42     & 134     & 0.04    & 111     & 2.49     & 77    & 1.47     & 30    \\
$ww(2)$    & 8   & 2      & 265.79   & 3671    & 14.30   & 2317    & 227.66   & 2140  & 30.58    & 237   \\
$ww(3)$    & 24  & 3      & $>$ 1h   &         & $>$ 1h  &         & $>$ 1h   &       & $>$ 1h   &       \\ \hline
$\Lmax$   & 3   & 1      & 0.01     & 54      & 0.01    & 54      & 3.58     & 54    &          &       \\
$\Lint$   & 5   & 2      & 0.59     & 478     & 0.17    & 478     & 83.26    & 478   &          &
\end{tabular}
}}
\caption{Running times and number of membership queries for the \nLStar{} algorithm. For \NLambda{} we used two version: \NLambda{}$^{\mathit{ord}}$ uses the total order symmetry \NLambda{}$^{\mathit{eq}}$ uses the equality symmetry.}
\label{learning_results}
\end{table}

\section{Related work}
\label{sec:related}

As stated in the introduction, \NLambda{} \cite{EPTCS207.3} and \LOIS{} \cite{kopczynski1716} use first-order formulas to represent nominal sets and use SMT solvers to manipulate them.
This makes both libraries very flexible and they indeed implement the equality symmetry as well as the total order symmetry.
As their representation is not unique, the efficiency depends on how the logical formulas are constructed.
As such, they do not provide complexity results.
In contrast, our direct representation allows for complexity results (Section~\ref{sec:impl}) and leads to different performance characteristics (Section~\ref{sec:evaluation}).

A second big difference is that both \NLambda{} and \LOIS{} implement a ``programming paradigm'' instead of just a library.
This means that they overload natural programming constructs in their host languages (Haskell and \cplusplus{} respectively).
For programmers this means they can think of infinite sets without having to know about nominal sets.

It is worth mentioning that an older (unreleased) version of \NLambda{} implemented nominal sets with orbits instead of SMT solvers \cite{towardsnominal}.
However, instead of characterising orbits (e.g., by its dimension), they represent orbits by a representative element.
The authors of \NLambda{} have reported that the current version is faster \cite{EPTCS207.3}.

The theoretical foundation of our work is the main representation theorem in \cite{bojanczyk2014automata}.
We add to that by instantiating it to the total order symmetry and distil a concrete representation of nominal sets.
As far as we know, we provide the first implementation of the representation theory in \cite{bojanczyk2014automata}.

Another tool using nominal sets is Mihda~\cite{ferrari2005}, where only the equality symmetry is implemented.
This tool translates the $\pi$-calculus to history-dependent automata (HD-automata), with the aim of minimisation and checking bisimilarity.
The implementation in OCaml is based on \emph{named sets}, which are finite representations for nominal sets.
The theory of named sets is well-studied and has been used to model various behavioural models with local names.
For those results, the categorical equivalences between named sets, nominal sets and a certain (pre)sheaf category are exploited~\cite{CianciaKM10, CianciaM10, Staton07}.
In particular, a finite representation similar to ours is presented in~\cite{CianciaKM10}, see Section~\ref{sec:rel-work-ciancia}.
The total order symmetry is not mentioned in their work.

Fresh OCaml \cite{shinwell2005fresh} and Nominal Isabelle \cite{urban2005nominal} are
both specialised in name-binding and $\alpha$-conversion used in proof systems.
They only use the equality symmetry and do not provide a library for manipulating nominal sets.
Hence they are not suited for our applications.

On the theoretical side, there are many complexity results for register automata \cite{DBLP:journals/corr/abs-1209-0680,DBLP:conf/lics/MurawskiRT15}.
In particular, we note that problems such as emptiness and equivalence are NP-hard depending on the type of register automaton.
This does not easily compare to our complexity results for minimisation.
One difference is that we use the total order symmetry, where the local symmetries are always trivial (Lemma~\ref{group_trivial}).
As a consequence, all the complexity required to deal with groups vanishes.
Rather, the complexity is transferred to the input of our algorithms, because automata over the equality symmetry require more orbits when expressed over the total order symmetry.
Another difference is that register automata allow for duplicate values in the registers.
In nominal automata, such configurations will be encoded in different orbits.
An interesting open problem is whether equivalence of unique-valued register automata is in {\scshape Ptime} \cite{DBLP:conf/lics/MurawskiRT15}.

Orthogonal to nominal automata, there is the notion of symbolic automata \cite{DAntoniV17,maler2017generic}.
These automata are also defined over infinite alphabets but they use predicates on transitions, instead of relying on symmetries.
Symbolic automata are finite state (as opposed to infinite state nominal automata) and do not allow for storing values.
However, they do allow for general predicates over an infinite alphabet, including comparison to constants.

\section{Conclusion and Future Work}
\label{sec:fw}

We presented a concrete finite representation for nominal sets over the total order symmetry.
This allowed us to implement a library, \ONS{}, and provide complexity bounds for common operations.
The experimental comparison of \ONS{} against existing solutions for automata minimisation and 
learning show that our implementation is much faster in many instances. 
As such, we believe \ONS{} is a promising implementation of nominal techniques.

A natural direction for future work is to consider other symmetries,
such as the equality symmetry. Here, we may take inspiration from existing tools such as Mihda (see Section~\ref{sec:related}).
Another interesting question is whether it is possible to translate a nominal automaton
over the total order symmetry which accepts an equality language to an
automaton over the equality symmetry.
This would allow one to efficiently move between symmetries.
Finally, our techniques can potentially be applied
to timed automata by exploiting the
intriguing connection between
the nominal automata that we consider and timed automata~\cite{BojanczykL12}.

\subsection*{Acknowledgements}
We would like to thank Szymon Toru\'{n}czyk and Eryk Kopczy\'{n}ski for their prompt help when using the \LOIS{} library.
For general comments and suggestions we would like to thank Ugo Montanari and Niels van der Weide.
At last, we want to thank the anonymous reviewers of ICTAC 2018 and 
this extended paper for their comments.

\bibliographystyle{elsarticle-num}
\bibliography{fastcalc}

\end{document}